\documentclass[12pt]{elsarticle}
\usepackage[T1]{fontenc}
\usepackage[a4paper, lmargin=13mm, rmargin=13mm, tmargin=18mm, bmargin=18mm]{geometry}
\usepackage{amssymb,amsthm,amsmath}
\usepackage{dsfont}
\usepackage{float}
\usepackage{tabularx}
\usepackage{xcolor}
\usepackage{enumitem}
\usepackage[small]{complexity}
\usepackage[skins,listings]{tcolorbox}
\usepackage{hyperref}
\usepackage{caption}
\usepackage{subcaption}
\usepackage[nameinlink,capitalise,noabbrev]{cleveref}
\usepackage[ruled,linesnumbered,vlined,noend]{algorithm2e}
\usepackage{tikz}
\usetikzlibrary {graphs,decorations.pathreplacing,calligraphy,graphs.standard,patterns.meta,calc,quotes,math,shapes.geometric,positioning}
\hypersetup{colorlinks = true, citecolor = blue, linkcolor = blue, breaklinks = true, urlcolor=blue}
\Crefname{figure}{Fig.}{Figs.}
\Crefname{equation}{Equation}{Equations}
\Crefname{theorem}{Theorem}{Theorems}
\Crefname{lemma}{Lemma}{Lemmas}
\Crefname{claim}{Claim}{Claims}
\crefname{algocf}{Algorithm}{Algorithms}
\def\ifdiff#1#2{\ifnum#1=#2\relax 0\else 1\fi}

\makeatletter
\def\ps@pprintTitle{%
  \let\@oddhead\@empty
  \let\@evenhead\@empty
  \let\@oddfoot\@empty
  \let\@evenfoot\@oddfoot}
\makeatother

\let\oldnl\nl
\newcommand{\nonl}{\renewcommand{\nl}{\let\nl\oldnl}}

\newcommand{\pname}[1]{\textcolor{black}{\textsc{#1}}}
\newtcolorbox[auto counter]{ProblemBox}[2][]{
	colframe=black!20!white,
	coltitle=black,
	arc=1.5mm,
	boxrule=.2mm,
	colbacktitle=white,
	colback=white,
	enhanced,
	adjusted title=flush left,
	attach boxed title to top center={yshift=-3mm,xshift=-3.5cm},
	boxed title style={colframe=white,righttitle=-3mm,lefttitle=-3mm},
	title=#2,#1
}

\newcommand{\problemClassic}[5]{
	\begin{ProblemBox}[label={#4},nameref={\pname{#5}}]{\pname{#1}}
		\begin{tabularx}{\textwidth}{rX}
			\textnormal{Input:} & \textnormal{#2} \\
                \textnormal{Question:} & \textnormal{#3}
		\end{tabularx}
		\vspace{-3,5mm}
	\end{ProblemBox}
}

\newcommand{\Nat}{\mathbb{N}}
\newtheorem{theorem}{Theorem}
\newtheorem{lemma}[theorem]{Lemma}
\newtheorem{claim}[theorem]{Claim}
\newtheorem{definition}[theorem]{Definition}

\DeclareMathOperator{\adj}{adj}
\newcommand{\qedclaim}{\hfill $\diamond$ \medskip}
\newenvironment{proofclaimfinal}{\noindent\emph{Proof of Claim.}}{}
\newenvironment{proofclaim}{\noindent\emph{Proof of Claim.}}{\qedclaim}

\journal{Discrete Applied Mathematics}

\begin{document}

\makeatletter
\DeclareRobustCommand*{\nameref}{%
\color{black}%
        \@ifstar\T@nameref\T@nameref
        }%
\makeatother

\begin{frontmatter}

\title{Parameterized complexity of the \texorpdfstring{$f$}{f}-Critical Set problem}

\author[UFCA]{Thiago Marcilon\corref{cor1}}
\ead{thiago.marcilon@ufca.edu.br}
\author[UFCA]{Murillo Inácio da Costa Silva}
\ead{murillo.inacio@aluno.ufca.edu.br}
\affiliation[UFCA]{organization={CCT, Universidade Federal do Cariri},
            city={Juazeiro do Norte},
            country={Brazil}}
\cortext[cor1]{Corresponding author}

\begin{abstract}
Given a graph~$G=(V,E)$, and a function~$f:V(G) \rightarrow \Nat$, an~$f$-reversible process on~$G$ is a dynamical system such that, given an initial vertex labeling~$c_0 : V(G) \rightarrow \{0,1\}$, every vertex~$v$ changes its label if and only if it has at least~$f(v)$ neighbors with the opposite label. The updates occur synchronously in discrete time steps $t=0,1,2,\ldots$. An $f$-critical set of $G$ is a subset of vertices of $G$ whose initial label is~$1$ such that, in an~$f$-reversible process on~$G$, all vertices reach label~$1$ within one time step and then remain unchanged. The critical set number~$r^c_f(G)$ is the minimum size of an $f$-critical set of $G$. Given a graph $G$, a threshold function $f$, and an integer $k$, the $f$-\textsc{Critical Set} problem asks whether~$r^c_f(G) \leq k$. We prove that this problem is $\NP$-complete for planar subcubic bipartite graphs with maximum threshold $m(f) = 2$ and $\W[1]$-hard when parameterized by the treewidth $tw(G)$ of $G$. Additionally, we show that the problem is $\FPT$ when parameterized by $tw(G)+m(f)$, $tw(G)+\Delta(G)$, and $k$, where $\Delta(G)$ denotes the maximum degree of~$G$. Finally, we present two kernels of sizes $O(k \cdot m(f))$ and $O(k \cdot \Delta(G))$.
\end{abstract}

\begin{keyword}
dynamical systems on graphs \sep threshold process \sep reversible process \sep critical set
\end{keyword}

\end{frontmatter}

\section{Introduction}
\label{sec:intro}

Let $G=(V,E)$ be a simple, finite, and undirected graph with~$n$ vertices and~$m$ edges. A \emph{discrete dynamical process on~$G$} is an infinite sequence $c = (c_{t})_{t \in \Nat} = (c_0, c_1, \dotsc)$ of \emph{configurations} $c_t : V(G) \rightarrow \{0,1\}$. We say that~$c_0$ is the \emph{initial configuration} of~$c$, and $c_t(v)$ denotes the \emph{state} of vertex~$v$ at time~$t\in \Nat$. The transition from~$c_t$ to~$c_{t+1}$ occurs under the same rule at each time step. This approach is employed to model real-world processes in a wide variety of areas such as social influence~\cite{Chen,French,Kempe,Poljak1}, gene expression networks~\cite{Huang}, immune systems~\cite{Agur}, cellular automata~\cite{Allouche}, marketing strategies~\cite{Dreyer,Kempe}, opinion and disease dissemination~\cite{opinion,disease-example}, simulations of biological cell populations~\cite{KnutsonThesis}, neural networks~\cite{Goles}, local interaction games~\cite{Montanari,Morris}, and distributed computing~\cite{Geurts,Roncato,Peleg}.

In this study, we consider a type of iterative process on simple graphs that generalizes the majority-voting approach studied by Peleg~\cite{Peleg}. In this type of process, the states of the vertices at each time step change according to a \textit{threshold function} $f:V(G) \rightarrow \Nat$: all vertices change their states simultaneously, and a vertex $v$ changes its state at time $t+1$ if and only if it has at least~$f(v)$ neighbors with the opposite state at time~$t \in \Nat$. Denoting the \textit{neighborhood} of~$v$ in $G$ by~$N_G(v)$ (we omit the subscript in the notation when the graph is clear from the context), the update rule for each vertex $v$ is formally defined as
\begin{equation}
c_{t+1}(v) = \begin{cases}
1-c_{t}(v) & \text{, if } |\{u \in N(v) \mid c_{t}(u) \neq c_{t}(v)\}| \geq f(v);\\
c_t(v) & \text{, otherwise.}
\end{cases}
\tag{R1}\label{eq:updRule}
\end{equation}

We say that~$f(v)$ is the \emph{threshold} of~$v$, which does not change during the process. Given an initial configuration~$c_0$, we define an~\emph{$f$-reversible process on~$G$}, denoted by~$(c_{t})_{t \in \Nat}$ (or simply~$c$ if the graph and threshold function are known in the context), as an iterative process that follows Rule~\eqref{eq:updRule}. If~$f$ is a~$k$-constant function, the process is called a \emph{$k$-reversible process on $G$}, where $k\in \Nat$. We also say that an~$f$-reversible process~$c$ on a graph~$G=(V,E)$ is \emph{started} by the vertex subset~$S = \{v \mid c_0(v)=1\}$.

An $f$-reversible process is analogous to the majority voting approach studied by Mustafa and Peke\v{c}~\cite{Mustafa} with $f(v)= \Bigl \lfloor \frac{d_G(v)}{2} \Bigr \rfloor +1$, where~$d_G(v)$ denotes the degree of~$v$ in the graph $G$ (we omit the subscript in the notation when the graph is clear from the context). Several studies have investigated iterative processes on graphs under the constraint that each vertex is allowed to change its state only from $0$ to $1$ during the process~\cite{Dreyer,Geurts,Kempe}. Such processes are referred to as \emph{$f$-irreversible processes} in a graph~$G$ with threshold function~$f$.

In an $f$-reversible process~$c$, due to \Cref{eq:updRule}, every vertex considers only its own state and the states of its neighbors to define its next state at each time step. Therefore, we consider only connected graphs in this study. Moreover, if~$f(v)>d(v)$ for some vertex~$v$, then~$c_0(v) = c_t(v)$, for all~$t > 0$. Hence, for each vertex $v$, we consider only threshold values in the set~$\{0, \ldots, d(v)+1\}$. Note that if~$f(v)=0$, then the state of~$v$ changes at every time step~$t \in \Nat$.

We say that an~$f$-reversible process on $G$ \emph{converges} at time $t \in \Nat$ if each vertex in $G$ has state~1 at every time $t' \geq t$, and there is at least one vertex with state~$0$ at every time $t' < t$. A set $S \subseteq V$ is called an $f$-\emph{conversion set} of~$G$ if the $f$-reversible process started by~$S$ eventually converges at some time $t$. The \emph{conversion set number}~$r_f(G)$ is the minimum size of an $f$-conversion set of~$G$. For any integer~$k > 0$, if~$f$ is a $k$-constant function, then $f$ may be replaced by~$k$ in all related terminology; for example, an $f$-conversion set can be called a $k$-conversion set, and the conversion set number can be denoted by~$r_k(G)$.

It is clear that~$r_1(G) = n$ for any graph~$G$. Dourado~\emph{et al.}~\cite{reversible2012} proved that determining whether $r_2(G) \leq k$ is $\NP$-hard, presented an upper bound for~$r_2(T)$ for trees~$T$ and proposed a quadratic dynamic programming algorithm to compute~$r_f(P)$ on paths~$P$ and threshold values in~$\{1,2\}$, but for restricted threshold functions, leaving open the general case. Lima~\emph{et al.}~\cite{lima2018} improved the hardness result of Dourado~\emph{et al.} by showing that the same problem is actually $\NP$-complete (not only $\NP$-hard) even for bipartite graphs of maximum degree~3 and threshold values in~$\{1,2,3\}$. Costa, Marcilon, and Lima~\cite{COSTA2023} presented a polynomial algorithm to compute~$r_f(P)$ on paths $P$ and any threshold function, improving the result for paths of Dourado~\emph{et al.}~\cite{reversible2012}. They also showed that the problem is $\W[1]$-hard when parameterized by the combination of $k$ and the treewidth of the input graph.

A concept closely related to $f$-conversion sets is that of an $f$-critical set. A set $S \subseteq V$ is called an $f$-\emph{critical set} of $G$ if the $f$-reversible process started by~$S$ converges at time $t \leq 1$. The \emph{critical set number}~$r^c_f(G)$ is the minimum size of an $f$-critical set of~$G$. We now introduce the decision problem that is the focus of this study.

\problemClassic{$f$-Critical Set}{A simple graph~$G$, a threshold function $f:V(G) \rightarrow \Nat$, and an integer~$k>0$.}{$r^c_f(G) \leq k$?}{prob:criticalSet}{Critical Set}

If $f(v) = 0$ for some vertex $v$, then $G$ has no $f$-critical set, since the state of $v$ changes at every time step. Therefore, from now on, let us assume that $1 \leq f(v) \leq d_G(v)+1$ for every vertex $v \in V(G)$. Also, if $G=(V,E)$ is a simple graph, let $tw(G)$ denote the treewidth of $G$, and $\Delta(G)$ denote the maximum degree of a vertex in $G$. For any $S \subseteq V$, let $G[S]$ denote the subgraph of $G$ induced by $S$. Let $m(f) = \max\limits_{v \in V} f(v)$ be the maximum threshold of a vertex of $G$, where $f$ is a threshold function of~$G$.

The critical set number has been extensively studied in the context of $f$-irreversible processes, given that it relates to the interval number (or geodetic number) of a graph in the $P_3$ convexity~\cite{penso2015}, the $k$-domination number~\cite{fink1985}, and the TSS-Domination in the context of the \textsc{Target Set Selection} problem~\cite{ARAUJO2023}. However, its investigation in the context of $f$-reversible processes is recent. Fernandes \emph{et al.}~\cite{fernandes2024freversible} presented a linear algorithm for the $f$-\nameref{prob:criticalSet} problem restricted to paths. Lima \emph{et al.}~\cite{lima2018} showed the $\NP$-completeness of the problem for subcubic bipartite graphs and $m(f) \leq 3$. Although their result is about the problem regarding $f$-conversion sets, their reduction also holds for the $f$-\nameref{prob:criticalSet} problem. To the best of our knowledge, beyond these two results, no complexity results or broader structural characterizations are known. This lack of general understanding motivated the present study.

Since the classical complexity of the problem is already intractable in very restricted graph classes, we follow a parameterized approach. The parameters we consider reflect both global and local structural aspects of the reversible processes. The value $k$ limits the size of the $f$-critical set, $tw(G)$ captures the global structure of the graph, and $m(f)$ and $\Delta(G)$ control the local constraints regarding the change in states. Therefore, these parameters are natural choices for studying the parameterized complexity of the $f$-\nameref{prob:criticalSet} problem.

\paragraph{Our contributions} In this paper, we establish new complexity results for the $f$-\nameref{prob:criticalSet} problem with respect to the parameters $k$, $tw(G)$, $m(f)$, and $\Delta(G)$.
\begin{itemize}
    \item In \Cref{sec:hardness}, we prove that the problem is $\NP$-complete even when restricted to planar subcubic bipartite graphs with maximum threshold~$2$, thus strengthening the result of Lima \emph{et al.}~\cite{lima2018};
    \item Also in \Cref{sec:hardness}, we prove that the problem is $\W[1]$-hard when parameterized by~$tw(G)$;
    \item In \Cref{sec:tractable}, we show that the problem is $\FPT$ when parameterized by $tw(G)+m(f)$ and $tw(G)+\Delta(G)$;
    \item In \Cref{sec:tractablek}, we present an $\FPT$ algorithm for the problem parameterized by~$k$;
    \item Also in \Cref{sec:tractablek}, we present a polynomial kernel for the parameters $k+m(f)$ and $k+\Delta(G)$.
\end{itemize}  
By combining hardness and tractability results, we delineate the complexity of the problem within the framework of Parameterized Complexity Theory with respect to the parameters $tw(G)$, $m(f)$, $\Delta(G)$ and $k$. Our results show that determining whether $r_f^c(G) \leq k$ is generally easier than determining whether $r_f(G) \leq k$: the latter is $\W[2]$-hard when parameterized by $k$, whereas the former is fixed-parameter tractable.

\paragraph{Complexity landscape}
\Cref{tab:results} summarizes the parameterized complexity landscape for the $f$-\nameref{prob:criticalSet} problem in both settings, reversible and irreversible processes, for the parameters $k, tw(G), m(f)$ and $\Delta(G)$, including our results.

\begin{table}[H]
    \caption{\label{tab:results}Summary of the parameterized complexity landscape for the $f$-\nameref{prob:criticalSet} problem.}
\[
\bgroup
\def\arraystretch{1.35}
\begin{array}{|c|c|c|}
\hline
\text{\textbf{Parameter}} & \text{\textbf{Reversible processes}} & \text{\textbf{Irreversible processes}} \\ \hline
k & \FPT\, (\text{\Cref{thm:fptk}}) & \W[2]\text{-hard~\cite{penso2015,ARAUJO2023,CAPPELLE2022}} \\
tw(G) & \W[1]\text{-hard}\, (\text{\Cref{thm:hardtw}}) & \text{open} \\
k+m(f) & \FPT\, (\text{\Cref{thm:fptkmf}}) & \W[2]\text{-hard~\cite{ARAUJO2023}} \\
k+\Delta(G) & \FPT\, (\text{\Cref{thm:fptkmf}}) & \FPT~\text{\cite{ARAUJO2023}}\\
k+tw(G) & \FPT\, (\text{\Cref{thm:fptk}}) & \FPT~\text{\cite{ARAUJO2023}}\\
m(f)+\Delta(G)& \text{para-}\NP\text{-comp.} (\text{\cite{lima2018}, \Cref{thm:paraplanarbip}}) & \text{para-}\NP\text{-comp.}~\text{\cite{penso2015}}\\
tw(G) + m(f) & \FPT\, (\text{\Cref{thm:fpttwmf}}) & \text{open} \\
tw(G) + \Delta(G) & \FPT\, (\text{\Cref{thm:fpttwmf}}) & \text{open} \\
\hline
\end{array}
\egroup
\]
\end{table}
\section{Hardness results for the parameters \texorpdfstring{$tw(G), m(f)$}{tw(G), m(f)} and \texorpdfstring{$\Delta(G)$}{Δ(G)}}
\label{sec:hardness}

In this section, we prove two hardness results for the $f$-\nameref{prob:criticalSet} problem. We show that it is $\NP$-complete even when restricted to planar subcubic bipartite graphs with maximum threshold~$2$, and that it is $\W[1]$-hard when parameterized by $tw(G)$.

\begin{theorem}
\label{thm:paraplanarbip}
$f$-\nameref{prob:criticalSet} is $\NP$-complete for planar subcubic bipartite graphs with maximum threshold~$2$.
\end{theorem}

\begin{proof}
The problem is clearly in $\NP$. To prove that the problem is $\NP$-hard, we show a reduction from \textsc{Planar Vertex Cover}, which is a known $\NP$-complete problem~\cite{garey1977}. It consists in determining whether a planar graph $G$ has a vertex cover of size at most $k$, that is, a set $S \subseteq V(G)$ such that each edge in $G$ has at least one endpoint in $S$. Let $(G=(V,E),k)$ be an instance of \textsc{Planar Vertex Cover}, and $(G'=(V',E'),f,k')$ be the constructed instance of $f$-\nameref{prob:criticalSet}.

The construction of $G'$ and $f$ is as follows:
\begin{itemize}
\item For each $e \in E$, add the vertex $e$ in $G'$ with threshold $2$. Let $F$ be the set of vertices added in this step;
\item For each $v \in V$, add a path with $2\Delta(G)-1$ vertices in $G'$ all of which have threshold $1$. Let $P_v = \{v_1,v_2,\ldots,v_{2\Delta(G)-1}\}$ be the set of vertices in the path $v_1,v_2,\ldots,v_{2\Delta(G)-1}$ and $P = \bigcup\limits_{v \in V} P_v$. Note that $P_v$ has $\Delta(G)$ vertices with odd indices: $v_1,v_3,\ldots,v_{2\Delta(G)-1}$. Let $P_v^e$ and $P_v^o$ be the sets of vertices in $P_v$ with even and odd indices, respectively. Furthermore, let $P^e = \bigcup\limits_{v \in V} P_v^e$ and $P^o = \bigcup\limits_{v \in V} P_v^o$;
\item For each $e=vu \in E$, add in $G'$ an edge between $e$ and $v_i$ and between $e$ and $u_{j}$ where $i$ and $j$ are any odd indices in a way that, after adding all such edges, in every set $P_v$, each vertex in $P_v$ has degree at most $3$ and degree $2$ if it is in $P_v^e$;
\item For each vertex $v_i$ in some $P_v$ that does not have a neighbor in $F$, add another vertex $v'_i$ in $V'$ with threshold $2$ and an edge between $v_i$ and $v'_i$. Let $Q$ be the set of all vertices added in this step; $Q_v^e$ be the set of vertices in $Q$ adjacent to a vertex in $P_v^e$; $Q_v^o$ be the set of vertices in $Q$ adjacent to a vertex in $P_v^o$; $Q^e = \bigcup\limits_{v \in V} Q_v^e$; and $Q^o = \bigcup\limits_{v \in V} Q_v^o$. After this last step, every vertex in any $P_v$ has degree $3$, except for $v_1$ and $v_{2\Delta(G)-1}$, which have degree $2$.
\end{itemize}

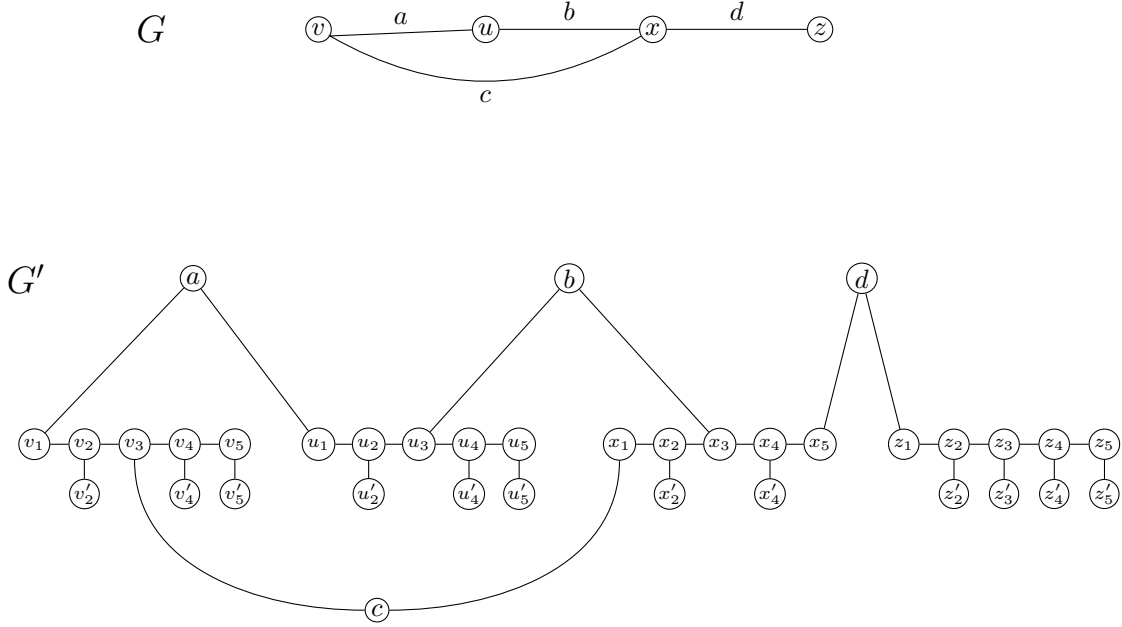
\begin{figure}[htb]
    \centering
    \begin{tikzpicture}[scale=1.1,auto,font=\scriptsize]
        \tikzset{vertex/.style={draw, circle, minimum size=4pt,inner sep=1pt}};

        \node () at (-5,4) {\Large{$G$}};
        \node () at (-6.5,1) {\Large{$G'$}};

        \node[vertex] (v) at (-3,4) {\normalsize{$v$}};
        \node[vertex] (u) at (-1,4) {\normalsize{$u$}};
        \node[vertex] (x) at (1,4) {\normalsize{$x$}};
        \node[vertex] (z) at (3,4) {\normalsize{$z$}};
        
        \draw[-] (x) to[bend left] node {\normalsize{$c$}} (v) to node {\normalsize{$a$}} (u) to node {\normalsize{$b$}} (x) to node {\normalsize{$d$}} (z);    
        \node[vertex] (a) at (-4.5,1) {\normalsize{$a$}};
        \node[vertex] (c) at (-2.3,-3) {\normalsize{$c$}};
        \node[vertex] (b) at (0,1) {\normalsize{$b$}};
        \node[vertex] (d) at (3.5,1) {\normalsize{$d$}};

        \node[vertex] (v1) at (-6.4,-1) {$v_1$};
        \node[vertex] (v2) at (-5.8,-1) {$v_2$};
        \node[vertex] (v3) at (-5.2,-1) {$v_3$};
        \node[vertex] (v4) at (-4.6,-1) {$v_4$};
        \node[vertex] (v5) at (-4,-1) {$v_5$};
        \node[vertex,inner sep=0pt] (v2l) at (-5.8,-1.6) {$v'_2$};
        \node[vertex,inner sep=0pt] (v4l) at (-4.6,-1.6) {$v'_4$};
        \node[vertex,inner sep=0pt] (v5l) at (-4,-1.6) {$v'_5$};
        \draw[-] (v2) -- (v2l);
        \draw[-] (v4) -- (v4l);
        \draw[-] (v5) -- (v5l);
        \draw[-] (v1) -- (v2) -- (v3) -- (v4) -- (v5);        

        \node[vertex] (u1) at (-3,-1) {$u_1$};
        \node[vertex] (u2) at (-2.4,-1) {$u_2$};
        \node[vertex] (u3) at (-1.8,-1) {$u_3$};
        \node[vertex] (u4) at (-1.2,-1) {$u_4$};
        \node[vertex] (u5) at (-0.6,-1) {$u_5$};
        \node[vertex,inner sep=0pt] (u2l) at (-2.4,-1.6) {$u'_2$};
        \node[vertex,inner sep=0pt] (u4l) at (-1.2,-1.6) {$u'_4$};
        \node[vertex,inner sep=0pt] (u5l) at (-0.6,-1.6) {$u'_5$};
        \draw[-] (u2) -- (u2l);
        \draw[-] (u4) -- (u4l);
        \draw[-] (u5) -- (u5l);
        \draw[-] (u1) -- (u2) -- (u3) -- (u4) -- (u5);
        
        \node[vertex] (x1) at (0.6,-1) {$x_1$};
        \node[vertex] (x2) at (1.2,-1) {$x_2$};
        \node[vertex] (x3) at (1.8,-1) {$x_3$};
        \node[vertex] (x4) at (2.4,-1) {$x_4$};
        \node[vertex] (x5) at (3,-1) {$x_5$};
        \node[vertex,inner sep=0pt] (x2l) at (1.2,-1.6) {$x'_2$};
        \node[vertex,inner sep=0pt] (x4l) at (2.4,-1.6) {$x'_4$};
        \draw[-] (x2) -- (x2l);
        \draw[-] (x4) -- (x4l);
        \draw[-] (x1) -- (x2) -- (x3) -- (x4) -- (x5);

        \node[vertex] (z1) at (4,-1) {$z_1$};
        \node[vertex] (z2) at (4.6,-1) {$z_2$};
        \node[vertex] (z3) at (5.2,-1) {$z_3$};
        \node[vertex] (z4) at (5.8,-1) {$z_4$};
        \node[vertex] (z5) at (6.4,-1) {$z_5$};
        \node[vertex,inner sep=0pt] (z2l) at (4.6,-1.6) {$z'_2$};
        \node[vertex,inner sep=0pt] (z3l) at (5.2,-1.6) {$z'_3$};
        \node[vertex,inner sep=0pt] (z4l) at (5.8,-1.6) {$z'_4$};
        \node[vertex,inner sep=0pt] (z5l) at (6.4,-1.6) {$z'_5$};
        \draw[-] (z2) -- (z2l);
        \draw[-] (z3) -- (z3l);
        \draw[-] (z4) -- (z4l);
        \draw[-] (z5) -- (z5l);
        \draw[-] (z1) -- (z2) -- (z3) -- (z4) -- (z5);

        \draw[-] (v1) -- (a) -- (u1);
        \draw[-] (u3) -- (b) -- (x3);
        \draw[-] (v3) to[in=180,out=270] (c);
        \draw[-] (c) to[in=270,out=0] (x1);
        \draw[-] (x5) -- (d) -- (z1);
    \end{tikzpicture}
    \caption{\label{fig:reduction2criticalSet} $G'$ is the graph resulting from the construction applied to the graph $G$.}
\end{figure}

This construction requires quadratic time. In fact, $|V'| = 2|V|(2\Delta(G)-1) - |E|$ and $|E'| \leq 3|V'|/2$. In \Cref{fig:reduction2criticalSet}, we present an example of this construction, in which graph $G'$ is the result of the construction applied to graph $G$.

The idea of this reduction is as follows: for any $f$-critical set $S$ of $G'$, each $P_v$ is either a subset of $S$ or disjoint from $S$, and for each $e = vw$, either $P_v$ or $P_w$ is subset of $S$. This ensures that the vertices in $G$ represented by the sets $P_v$ in an $f$-critical set of $G'$ form a vertex cover of $G$.

We have that $m(f) \leq 2$ because each vertex has threshold $2$ if it is in $F \cup Q$ and threshold $1$ otherwise. In addition, $G'$ is subcubic because each vertex in $F$ has degree $2$, each vertex in $Q$ has degree $1$, and each vertex in every path $P_v$ has degree $2$ or $3$. $G'$ is bipartite because $V'$ can be partitioned into $P^e \cup F \cup Q^o$ and $P^o \cup Q^e$.

Furthermore, $G'$ is planar because starting at a planar embedding of $G$, we can subdivide each of its edges to obtain the vertices in $F$ and transform each vertex $v \in V$ in the subgraph $G[P_v \cup Q_v]$.

This transformation can be done in the following manner, as depicted in \Cref{fig:planarlittlevertices}: add the vertices in $P_v$ alongside where the circumference of $v$  is drawn, positioning each vertex $v_k$ in $P_v^o$ precisely where each incident edge touches $v$ if $v_k$ has a neighbor in $F$. The vertices in $Q_v$ can be positioned inside the circumference of $v$ near its sole neighbor in $P_v$. Thus, we obtain a planar embedding for~$G'$.

\begin{figure}[htb]
    \centering
    \begin{tikzpicture}[scale=0.85,auto,font=\scriptsize]
        \tikzset{vertex/.style={draw, circle, minimum size=10pt,inner sep=1pt}};

        \node[vertex] (v1) at (60:3){$v_1$};
        \node[vertex] (v2) at (112:3){$v_2$};
        \node[vertex,inner sep=0pt] (v2l) at (112:2){$v'_2$};
        \node[vertex] (v3) at (165:3){$v_3$};
        \node[vertex] (v4) at (217:3){$v_4$};
        \node[vertex,inner sep=0pt] (v4l) at (217:2){$v'_4$};
        \node[vertex] (v5) at (270:3){$v_5$};
        \node[vertex,inner sep=0pt] (v5l) at (270:2){$v'_5$};
        \node (v) at (0,0) {\Huge{$v$}};

        \draw[circle,color=black!20] (-85.5:3) arc (-85.5:55.5:3);
        \draw[circle] (65:3) arc (65:107:3);
        \draw[circle] (117:3) arc (117:160:3);
        \draw[circle] (170:3) arc (170:212:3);
        \draw[circle] (222:3) arc (222:265:3);

        \draw[-,color=black!20] (v1) -- (60:4.5);
        \draw[-] (v2) -- (v2l);
        \draw[-,color=black!20] (v3) -- (165:4.5);
        \draw[-] (v4) -- (v4l);
        \draw[-] (v5) -- (v5l);
    \end{tikzpicture}
    \caption{\label{fig:planarlittlevertices} The vertices $v_1$ and $v_3$ are positioned in the circumference of $v$ precisely where the incident edges touch $v$, assuming $d_G(v) = 2$ and $\Delta(G)=3$. The other vertices are positioned between or after them. In black, we have the subgraph $G'[P_v \cup Q_v]$.}
\end{figure}
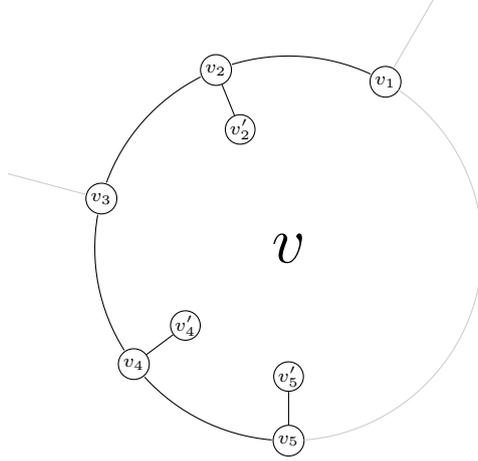

We finish by proving the following claim.

\begin{claim}
\label{claim:equivmfD}
$G$ has a vertex cover of size at most $k$ if and only if $G'$ has an $f$-critical set of size at most $k' = |F| + |Q| + k \cdot (2\Delta(G)-1)$.
\end{claim}

\begin{proofclaimfinal}
Assume that $G=(V,E)$ has a vertex cover $S$ of size at most $k$. Let $S' = F \cup Q \cup \bigcup\limits_{v \in S} P_v$. It follows that $|S'| = |F| + |Q| + |S| \cdot (2\Delta(G)-1) \leq k'$. Furthermore, $S'$ is an $f$-critical set of $G'$ since each vertex in $F \cup Q$ has threshold $2$ and has at most one neighbor in $V' \setminus S'$; all neighbors of each vertex in $S' \setminus (F \cup Q)$ are also in $S'$; and each vertex in $V' \setminus S'$ has at least one neighbor in $S'$ and threshold $1$.

Conversely, assume that $G'$ has an $f$-critical set $S'$ of size at most $k'$, and let $c$ be the $f$-reversible process on $G'$ started by $S'$. Since $f(v) > d(v)$ for each vertex in $Q$, $Q \subseteq S'$. Assume that $v \notin S'$ for some $v \in F$. On the one hand, its two neighbors in $P$ are not in $S'$ because both of them have threshold equal to $1$ and, otherwise, their states would be $0$ at time $1$ in $c$. On the other hand, its two neighbors in $P$ are in $S'$ because $v$ itself has threshold $2$ and, otherwise, its state would still be $0$ at time $1$ in $c$, a contradiction. It follows that $F \subseteq S'$.

Furthermore, for some $P_v$, if $u',w' \in P_v$ are two vertices such that $u' \in S'$ and $w' \notin S'$, then there are two neighbors $u,w \in P_v$ such that $u \in S'$ and $w \notin S'$, since $G'[P_v]$ is connected. This implies that $u$ has state $0$ at time $t=1$ in $c$ since $f(u) = 1$, a contradiction. Therefore, there is a vertex $v \in P_v \cap S'$ if and only if $P_v \subseteq S'$. 

With this, we can conclude that $S' = F \cup Q \cup \bigcup\limits_{v \in S} P_v$ for some $S \subseteq V$. Since $|S'| \leq k' = |F| + |Q| + k \cdot (2\Delta(G)-1)$ and each set $P_v$ has size $2\Delta(G)-1$, it follows that $|S| \leq k$. Furthermore, since $S'$ is an $f$-critical set of $G'$ and each vertex in $F$ has threshold $2$, each vertex $e \in F$ must have at least one neighbor in $P \cap S'$, say $v_k \in P_v$, which implies that $v \in S$. This means that for each edge $e$ in $G$, there is a vertex $v \in S$ such that $e$ is incident to $v$. Therefore, $S$ is a vertex cover of~$G$.
\end{proofclaimfinal}
\end{proof}

This result shows that the problem is intractable even for planar subcubic bipartite graphs with maximum threshold~$2$, which means that it is para-$\NP$-complete when parameterized by $m(f)+\Delta(G)$. Consequently, the problem is not in $\FPT$, not even in $\XP$, with respect to these parameters, unless $\P = \NP$.

The next theorem shows that the problem is parameterized intractable for the parameter $tw(G)$, which means that it is also not in $\FPT$, unless $\FPT = \W[1]$.

\begin{theorem}
\label{thm:hardtw}
$f$-\nameref{prob:criticalSet} is $\W[1]$-hard when parameterized by $tw(G)$.
\end{theorem}

\begin{proof}
We show a parameterized reduction from the \textsc{Clique} problem. Given a graph $G$ and integer $k$, the problem asks whether there is a clique in $G$ of size at least $k$. This is a known $\W[1]$-complete problem when parameterized by $k$. Let $(G=(V,E),k)$ be an instance of \textsc{Clique}, where $n = |V|$ and $m = |E|$, and $(G'=(V',E'),f,k')$ be the instance of $f$-\nameref{prob:criticalSet} that will be constructed.

Let $V = \{v_1,v_2,\ldots,v_n\}$ and $q = 2n$. For each $v_i \in V$, let $M_i = n+i-1$ and $N_i = n-i+1$. Note that $M_i+N_i = q$ for any $1 \leq i \leq n$, and for $1 \leq j \neq i \leq n$, either $M_i+N_j < q$ or $M_j+N_i < q$.

The construction of $G'$ and $f$ is as follows:
\begin{itemize}
    \item For each $1 \leq r \leq k$ and $1 \leq i \leq n$, add to $G'$ the set $U^r_i$ of $q$ vertices, all with threshold $k+1$, that induces a cycle. Let $U^r = \bigcup\limits_{i=1}^n U^r_i$ and $U = \bigcup\limits_{r=1}^k U^r$;
    
    \item For each $1 \leq r \leq k$, add to $G'$ the vertex $u^r$ with threshold $q$ and an edge between $u^r$ and each vertex in $U^r$;
    
    \item For each $1 \leq r \neq s \leq k$, add to $G'$ the set of vertices $C^{r,s} = \{a^{r,s},b^{r,s}\}$ both vertices with threshold~$q$. For any $1 \leq r \leq k$, let $C^r = \bigcup\limits_{s=1}^{r-1} C^{r,s} \cup \bigcup\limits_{s=r+1}^k C^{r,s}$ and $C = \bigcup\limits_{r=1}^k C^r$;
    
    \item For every $1 \leq r \neq s \leq k$ and $1 \leq i \leq n$, add to $G'$ an edge between $M_i$ vertices in $U^r_i$ and $a^{r,s}$, and an edge between the other $N_i$ vertices in $U^r_i$ (the ones that are not adjacent to $a^{r,s}$) and $b^{r,s}$. This must be done such that, for every $1 \leq r \leq k$, each vertex in $U^r$ has degree exactly $k+1$ and is either adjacent to $a^{r,s}$ for each $1 \leq s \neq r \leq k$, or adjacent to $b^{r,s}$ for each $1 \leq s \neq r \leq k$;
    
    \item For each $1 \leq r < s \leq k$ and $v_iv_j \in E$, add the set of vertices $Y^{r,s}_{i,j}$ and $Y^{r,s}_{j,i}$, each inducing a cycle with $2q$ vertices. Each of their vertices has threshold $3$. Let $Y^{r,s} = \bigcup\limits_{v_iv_j \in E} Y^{r,s}_{i,j} \cup Y^{r,s}_{j,i}$ and $Y = \bigcup\limits_{s=2}^k\bigcup\limits_{r=1}^{s-1} Y^{r,s}$;

    \item For each $1 \leq r < s \leq k$, add to $G'$ the vertex $y^{r,s}$ with threshold $2q$ and adjacent to every vertex in $Y^{r,s}$;
    
    \item For each $1 \leq r \neq s \leq k$ and each $v_iv_j \in E$, add an edge in $G'$ between each vertex in $Y^{r,s}_{i,j}$ ($Y^{r,s}_{j,i}$) and exactly one of the four vertices $a^{r,s},b^{r,s},a^{s,r},b^{s,r}$ in the following manner: add an edge between $N_i$ of its vertices and $a^{r,s}$ ($a^{s,r}$); add an edge between $M_i$ of its vertices and $b^{r,s}$ ($b^{s,r}$); add an edge between $N_j$ of its vertices and $a^{s,r}$ ($a^{r,s}$); and add an edge between $M_j$ of its vertices and $b^{s,r}$ ($b^{r,s}$);

    \item Let $k' = kq + \frac{k(k-1)}{2} \cdot 2q + (k+1) = k^2q+k+1$. Add to $G'$ the set of $k'$ vertices $W$, all with threshold $1$. Add to $G'$ an edge between each vertex in $W$ and each vertex in $C$. Additionally, for each $1 \leq r < s \leq k$, add to $G'$ an edge between each vertex in $W$ and $u^r$ and between each vertex in $W$ and $y^{r,s}$.
    
    \item Add to $G'$ the set of $k+1$ vertices $Z$. Add to $G'$ an edge between each vertex in $Z$ and each vertex in $U \cup Y \cup W$. Let the threshold of each vertex in $Z$ be $|U| + |Y| + |W| + 1$, that is, its degree in $G'$ plus one.
\end{itemize}

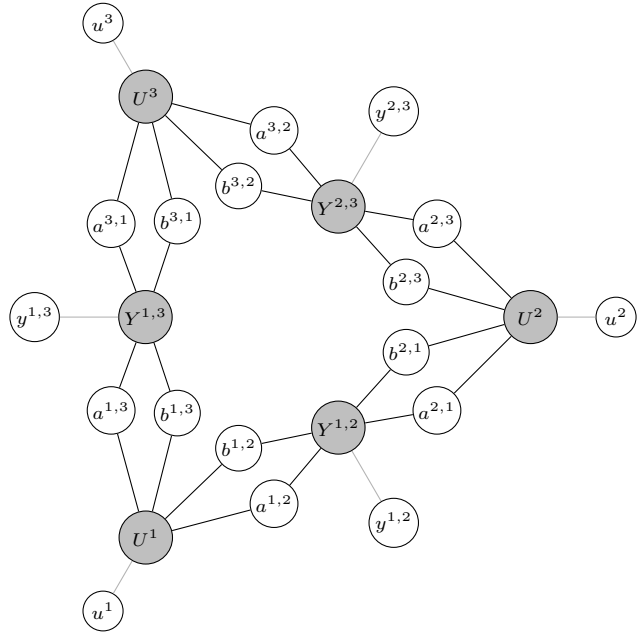
\begin{figure}[!htb]
    \centering
    \begin{tikzpicture}[scale=0.9,font=\scriptsize]
        \tikzset{vertex/.style={draw, circle, minimum size=15pt,inner sep=1pt}};
        \tikzset{setU/.style={draw, circle, minimum size=20pt,inner sep=1pt,fill=black!25}};
        \tikzset{setY/.style={draw, circle, minimum size=20pt, inner sep=1pt,fill=black!25}};
        \newcounter{count}
        \pgfmathsetmacro{\diam}{4}
        \foreach \r/\t in {1/-1,2/0,3/1}{
            \begin{scope}[rotate=120*\t]
                \node[vertex] (u\r) at (0:\diam+1) {$u^\r$};
                \node[setU] (U\r) at (0:\diam-0.25) {$U^\r$};
                \draw[-,dashed] (u\r)  to (U\r);
                \ifnum\r=2\setcounter{count}{-3}\else\setcounter{count}{3}\fi;
                \foreach \s in {1,2,3}{
                    \tikzmath{
                        if \s != \r then {
                            {\ifnum\r=2\addtocounter{count}{2}\else\addtocounter{count}{-2}\fi;
                            \node[vertex] (a\r\s) at (\thecount*30:\diam-1.25) {$a^{\r,\s}$};
                            \node[vertex] (b\r\s) at (\thecount*15:\diam-2) {$b^{\r,\s}$};
                            \draw[-] (a\r\s)  to (U\r);
                            \draw[-] (b\r\s)  to (U\r);};
                        };
                    }
                }
            \end{scope}
        }
        \foreach \r/\t/\s in {1/-1/2,2/0/3,1/1/3}{
            \begin{scope}[rotate=120*\t]
                \node[vertex] (y\r\s) at (60:\diam-0.5) {$y^{\r,\s}$};
                \node[setY] (Y\r\s) at (60:\diam-2.12) {$Y^{\r,\s}$};
                \draw[-,dashed] (Y\r\s)  to (y\r\s);
                \draw[-] (Y\r\s)  to (a\r\s);
                \draw[-] (Y\r\s)  to (a\s\r);
                \draw[-] (Y\r\s)  to (b\r\s);
                \draw[-] (Y\r\s)  to (b\s\r);
            \end{scope}
        }
    \end{tikzpicture}
    \caption{\label{fig:reductiontreewidth} The overall structure of $G'$ resulting from the construction applied to the path $P = v_1,v_2,v_3$ and $k=3$, with $Z \cup W$ omitted. Unlike dashed edges, a solid edge between an individual vertex $x$ (white colored nodes) and a set of vertices $Q$ (gray colored nodes) does not necessarily represent an edge between $x$ and every vertex in $Q$.}
\end{figure}

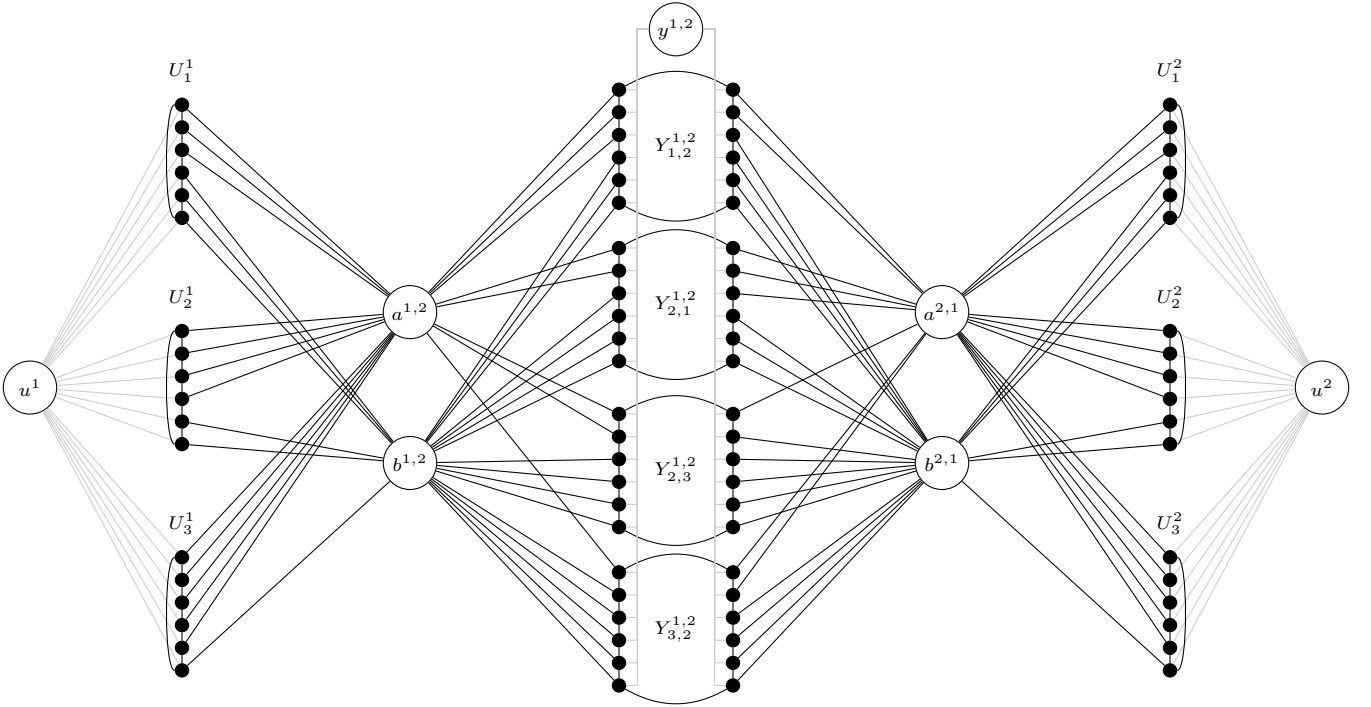
\begin{figure}[htb!]
    \centering
    \begin{tikzpicture}[scale=1,auto,font=\scriptsize]
        \tikzset{vertex/.style={draw, circle, minimum size=20pt,inner sep=1pt}};
        \tikzset{cvertex/.style={draw, circle, minimum size=5pt,inner sep=1pt,fill=black!25}};
        \node[vertex] (y12) at (5.5,4.75) {$y^{1,2}$};
        \foreach \i/\x/\v/\w in {1/-3.2/3/2,2/-1.1/2/3,3/1.1/2/1,4/3.2/1/2}{
            \foreach \j/\p/\jd in {1/0.75/12,2/0.45/11,3/0.15/10,4/-0.15/9,5/-0.45/8,6/-0.75/7}{
                
                \node[cvertex,fill=black] (y12\i\j) at (4.75,\x+\p) {};
                \node[cvertex,fill=black] (y12\i\jd) at (6.25,\x+\p) {};
            }
            \node () at (5.5,\x) {$Y^{1,2}_{\v,\w}$};
        }
        \foreach \i in {1,2,3,4}{
            \foreach \j/\jp in {1/2,2/3,3/4,4/5,5/6,7/8,8/9,9/10,10/11,11/12}{
                \draw[-] (y12\i\j) to (y12\i\jp);
            }
            \draw[-] (y12\i6) to[bend right] (y12\i7);
            \draw[-] (y12\i12) to[bend right] (y12\i1);
        }
        \foreach \l/\n/\m in {1/2/-1,2/1/1}{
            \node[vertex] (u\l) at (5.5+8.5*\m,0) {$u^\l$};
            \node[vertex] (a\l\n) at (5.5+3.5*\m,1) {$a^{\l,\n}$};
            \node[vertex] (b\l\n) at (5.5+3.5*\m,-1) {$b^{\l,\n}$};
            \foreach \i/\x in {1/3,2/0,3/-3}{
                \node () at (5.5+6.5*\m,\x+1.2) {$U^\l_\i$};
                \foreach \j/\p in {1/0.75,2/0.45,3/0.15,4/-0.15,5/-0.45,6/-0.75}{
                    \node[cvertex,fill=black] (q\l\i\j) at (5.5+6.5*\m,\x+\p) { };
                    \draw[-,color=black!20] (q\l\i\j) to (u\l);
                }
            }
            \foreach \i in {1,2,3}{
                \foreach \j/\jp in {1/2,2/3,3/4,4/5,5/6}{
                    \draw[-] (q\l\i\j) to (q\l\i\jp);
                }
                \draw[-] (q\l\i6) to[in=90-90*\m,out=90-90*\m,looseness=0.25] (q\l\i1);
            }
            \draw[-] (a\l\n) to (q\l11);
            \draw[-] (a\l\n) to (q\l12);
            \draw[-] (a\l\n) to (q\l13);
            \draw[-] (b\l\n) to (q\l14);
            \draw[-] (b\l\n) to (q\l15);
            \draw[-] (b\l\n) to (q\l16);

            \draw[-] (a\l\n) to (q\l21);
            \draw[-] (a\l\n) to (q\l22);
            \draw[-] (a\l\n) to (q\l23);
            \draw[-] (a\l\n) to (q\l24);
            \draw[-] (b\l\n) to (q\l25);
            \draw[-] (b\l\n) to (q\l26);

            \draw[-] (a\l\n) to (q\l31);
            \draw[-] (a\l\n) to (q\l32);
            \draw[-] (a\l\n) to (q\l33);
            \draw[-] (a\l\n) to (q\l34);
            \draw[-] (a\l\n) to (q\l35);
            \draw[-] (b\l\n) to (q\l36);
        }
        \foreach \i/\j/\k in {4/1/2,3/2/1,2/2/3,1/3/2}{
                \newcount\limpi
                \newcount\limpj
                \tikzmath{
                    \limpi = 3-\j+2;
                    \limpj = 6+3+\k;
                }
                \pgfmathsetmacro{\limi}{3-\j+1}
                \pgfmathsetmacro{\limj}{6+3+\k-1}
                \foreach \x in {1,...,\limi}{
                    \draw[-] (a12) to (y12\i\x);
                    \draw[-,color=black!20] (y12) -- ($(y12.west)-(0.15,0)$) |- (y12\i\x);
                }
                \foreach \x in {\the\limpi,...,6}{
                    \draw[-] (b12) to (y12\i\x);
                    \draw[-,color=black!20] (y12) -- ($(y12.west)-(0.15,0)$) |- (y12\i\x);
                }
                \foreach \x in {7,...,\limj}{
                    \draw[-] (b21) to (y12\i\x);
                    \draw[-,color=black!20] (y12) -- ($(y12.east)+(0.15,0)$) |- (y12\i\x);
                }
                \foreach \x in {\the\limpj,...,12}{
                    \draw[-] (a21) to (y12\i\x);
                    \draw[-,color=black!20] (y12) -- ($(y12.east)+(0.15,0)$) |- (y12\i\x);
                }
            }
    \end{tikzpicture}
    \caption{\label{fig:reductiontreewidthparts} Part of \Cref{fig:reductiontreewidth} in more detail. It depicts the vertices in $U^1$, $U^2$, $C^{1,2}$, $C^{2,1}$ and $Y^{1,2}$, the vertices $u^1$, $u^2$ and $y^{1,2}$ and the edges between them. For clarity, some edges are in light gray.}
\end{figure}

This construction requires quadratic time. In \Cref{fig:reductiontreewidth}, we provide an example of the construction applied to the path $P = v_1,v_2,v_3$. The gray nodes represent a set of vertices that induces cycles, whereas the white nodes represent regular vertices. A dashed edge between a vertex $x$ and a set of vertices $Q$ represents an edge between $x$ and every vertex in $Q$, whereas a solid edge between them does not necessarily represent an edge between $x$ and every vertex in $Q$. For clarity, in the figure, we omit the vertices in sets $Z$ and $W$.

In \Cref{fig:reductiontreewidthparts}, we have part of \Cref{fig:reductiontreewidth} in more detail. It depicts the vertices in $U^1$, $U^2$, $C^{1,2}$, $C^{2,1}$ and $Y^{1,2}$, the vertices $u^1,u^2$ and $y^{1,2}$ and the edges between these vertices. Note that $u^1$ is adjacent to every vertex in $U^1$, $u^2$ is adjacent to every vertex in $U^2$, and $y^{1,2}$ is adjacent to every vertex in $Y^{1,2}$.

The set $Q = C \cup Z \cup \{u^r \mid 1 \leq r \leq k\} \cup \{y^{r,s} \mid 1 \leq r < s \leq k\}$ (represented by the white nodes in \Cref{fig:reductiontreewidth} plus the vertices in $Z$) separates $G'$. All connected components of $G' - Q$ are either cycles (represented by the gray nodes in \Cref{fig:reductiontreewidth}) or isolated vertices in $W$. Since a cycle has a treewidth equal to $2$, it follows that $tw(G') \leq |Q| + 2 = (|C| + |Z| + |\{u^r \mid 1 \leq r \leq k\}| + |\{y^{r,s} \mid 1 \leq r < s \leq k\}|) + 2 = 2k(k-1) + (k+1) + k + \frac{k(k-1)}{2} + 2 = \frac{5k^2-k}{2} + 3$. Thus, $tw(G')$ is bounded by a quadratic function of $k$.

The idea of the reduction is as follows. If $G$ has a clique of size at least $k$, let $S$ be the union of $Z$, with $U^i_v$ for each $v$ in the clique (choosing distinct values for $i$ for each one), with $Y^{i,j}_{v,w}$ for each edge $vw$ in the clique such that $U^i_v$ and $U^j_w$ are in $S$. We claim that $S$ is an $f$-critical set of $G'$ of size at most $k'$.

Conversely, any $f$-critical set $S'$ of $G'$ of size at most $k'$ must consist of exactly the following vertices: all the vertices in $Z$; for each $1 \leq r \leq k$, all the vertices in the set $U^r_i$ for some $i$; and, for each pair of integers $1 \leq r < s \leq k$, all the vertices in the set $Y^{r,s}_{i,j}$ for some $i$ and $j$. Furthermore, because of how the values $M_i$ and $N_i$ are chosen, in order to ensure that every vertex $v \in C$ has at least $f(v)$ neighbors in $S'$, we must have that, for every $1 \leq r < s \leq k$, if $U^r_i \subseteq S'$ and $U^s_j \subseteq S'$, then $Y^{r,s}_{i,j} \subseteq S'$. This is the only way in which we can guarantee that each vertex $v \in C$ has at exactly $f(v)$ neighbors in $S'$. Because the set $Y^{r,s}_{i,j}$ only exists in $G'$ if $v_iv_j$ is an edge in $G$, this means that the vertices $v$ that represent each set $U^r_v$ containing in $S'$ form a clique of size $k$ in $G$.

The two following results are important in order to prove the equivalence of the two instances. Remember that $k' = k^2q+k+1$. Let $X = C \cup \{u^r \mid 1 \leq r \leq k\} \cup \{y^{r,s} \mid 1 \leq r < s \leq k\}$.

\begin{claim}
\label[claim]{claim:XWSl}
Let $S'$ be an $f$-critical set of $G'$ of size at most $k'$. Then, $X \cup W$ is disjoint from $S'$.
\end{claim}

\begin{proofclaim}
For every $v \in X$, $f(v) \leq d_{G'}(v) - k'$ (remember that every vertex in $X$ is adjacent to every vertex in $W$ and $|W| = k'$), since
\begin{itemize}
    \item for each $u^r$, $f(u^r) = k+1$ and $d(u^r) = k' + nq$;
    \item for each $y^{r,s}$, $f(y^{r,s}) = 2q$ and $d(y^{r,s}) = k' + 4mq$; and
    \item for each $v \in C$, $f(v) = q$ and $d(v) \geq k' + 2q$ (we can assume $m \geq 1$).
\end{itemize}

Assume that there is a vertex $v \in X$ in $S'$. This means that, even if every other vertex in $S'$ is in the neighborhood of $v$, $v$ would still have at least $f(v)$ neighbors in $V' \setminus S'$, which implies that the state of $v$ would change to $0$ at time $1$ in $c$. Therefore, $X$ and $S'$ are disjoint. This implies that $S'$ and $W$ are also disjoint as each vertex in $W$ has threshold $1$ and is adjacent to every vertex in $X$. We can then conclude that the set $X \cup W$ is disjoint from $S'$.
\end{proofclaim}

\begin{claim}
\label[claim]{claim:UYchoice}
Let $S'$ be an $f$-critical set of $G'$ of size at most $k'$. For every $1 \leq r \neq s \leq k$ and $1 \leq i \neq j \leq n$, either $U^r_i \subseteq S'$, or $U^r_i$ is disjoint from $S'$; and either $Y^{r,s}_{i,j} \subseteq S'$, or $Y^{r,s}_{i,j}$ is disjoint from $S'$.
\end{claim}

\begin{proofclaim}
Every vertex $v \in U^r_i$ has $f(v)-1 = k$ neighbors in $C \subseteq X$. Furthermore, by \Cref{claim:XWSl}, $X$ is disjoint from $S'$, which implies that $C$ also is. Thus, if there are two neighbors $v \in S',w \notin S'$ in some set $U^r_i$, then $c_1(v)=0$. Hence, since $G'[U^r_i]$ is connected, for every $1 \leq r \leq k$ and $1 \leq i \leq n$, either $U^r_i \subseteq S'$, or $U^r_i$ is disjoint from $S'$.

The same reasoning applies to the sets $Y^{r,s}_{i,j}$. Every vertex $v \in Y^{r,s}_{i,j}$ has at least $f(v)-1 = 2$ neighbors in $C \cup \{y^{r,s}\} \subseteq X$. Thus, if there are two neighbors $v \in S',w \notin S'$ in some set $Y^{r,s}_{i,j}$, then $c_1(v)=0$. Hence, since $G'[Y^{r,s}_{i,j}]$ is connected, for each set $Y^{r,s}_{i,j}$, either $Y^{r,s}_{i,j} \subseteq S'$, or $Y^{r,s}_{i,j}$ is disjoint from~$S'$.
\end{proofclaim}

We conclude by proving the equivalence of the instances.
\begin{claim}
$G$ has a clique of size at least $k$ if and only if $G'$ has an $f$-critical set of size at most $k'$.
\end{claim}

\begin{proofclaimfinal}
Assume that $G=(V,E)$ has a clique $S = \{v_{d_1},v_{d_2},\ldots,v_{d_k}\}$ of size at least $k$. Without loss of generality, let $d_i = i$ for every $1 \leq i \leq k$. Let $$S' = Z \cup \bigcup\limits_{r=1}^k U^r_r \cup \bigcup\limits_{s=2}^k \bigcup\limits_{r=1}^{s-1} Y^{r,s}_{r,s}.$$ Let $c$ be an $f$-reversible process on $G'$ started by $S'$. For each vertex $v \in Z$, $f(v) > d_{G'}(v)$, so $c_1(v)=1$. In addition, every vertex $v \in U^r_r$ has exactly $f(v)-1=k$ neighbors outside of $S'$, which are $u^r$ and the other $k-1$ vertices, either each $a^{r,s}$ or each $b^{r,s}$. Additionally, every vertex $v \in Y^{r,s}_{r,s}$ also has exactly $f(v)-1=2$ neighbors outside of $S'$, which are $y^{r,s}$ and exactly one of the vertices in the set $\{a^{r,s},b^{r,s},a^{s,r},b^{s,r}\}$. Thus, for every $v \in S'$, $c_1(v)=1$.

Every $u^r$ has exactly $f(u^r)=q$ neighbors in $S'$, which are the vertices in $U^r_r$. Every $y^{r,s}$ also has exactly $f(y^{r,s})=2q$ neighbors in $S'$, which are the vertices in $Y^{r,s}_{r,s}$. Every vertex $v$ in $(U \cup Y) \setminus S'$ and in $W$ has exactly $k+1 \geq f(v)$ vertices in $S'$, which are the vertices in $Z$.

Additionally, each $a^{r,s}$ has exactly $f(a^{r,s})=q$ neighbors in $S'$, which are $M_r$ neighbors in $U^r_r$ plus $N_r$ neighbors in either $Y^{r,s}_{r,s}$ if $r < s$, or $Y^{s,r}_{s,r}$ otherwise (remember that $M_i+N_i = q$ for any $1 \leq i \leq n$). Similarly, each $b^{r,s}$ has $f(b^{r,s})=q$ neighbors in $S'$, which are $N_r$ neighbors in $U^r_r$ plus $M_r$ neighbors in either $Y^{r,s}_{r,s}$ if $r < s$, or $Y^{s,r}_{s,r}$ otherwise. Thus, for every vertex $v \in V' \setminus S'$, $c_1(v)=1$. Therefore, $S'$ is an $f$-critical set of $G'$ of size at most~$k'$.

Conversely, let $S'$ be an $f$-critical set of $G'$ of size at most $k'$, and let $c$ be the $f$-reversible process on $G'$ started by $S'$. Since $f(v) > d_{G'}(v)$ for every $v \in Z$, it follows that $Z \subseteq S'$.

By \Cref{claim:UYchoice}, for every $1 \leq r \neq s \leq k$ and $1 \leq i \neq j \leq n$, either $U^r_i \subseteq S'$, or $U^r_i$ is disjoint from $S'$; and either $Y^{r,s}_{i,j} \subseteq S'$, or $Y^{r,s}_{i,j}$ is disjoint from $S'$. Furthermore, by \Cref{claim:XWSl}, $X \cup W$ is disjoint from $S'$. Notice that, since for each $1 \leq r \neq s \leq k$, $u^r,y^{r,s} \in X$, we have that $u^r,y^{r,s} \notin S'$.

Therefore, for each $1 \leq r \leq k$, there is at least one value $1 \leq i \leq n$ such that $U^r_i \subseteq S'$, since $u^r$ changes its state to $1$, and $N(u^r) \setminus W = U^r$. Also, for each $1 \leq r < s \leq k$, there is at least one pair $1 \leq i,j \leq n$ such that $Y^{r,s}_{i,j} \subseteq S'$, since $y^{r,s}$ changes its state to $1$, and $N(y^{r,s}) \setminus W = Y^{r,s}$.

Since $k' = k^2q+k+1 = |Z| + kq + \frac{k(k-1)}{2} \cdot 2q$, we can conclude that for each $1 \leq r \leq k$, there is \emph{exactly} one value $1 \leq i \leq n$ such that $U^r_i \subseteq S'$, and for each pair $1 \leq r < s \leq k$, there is \emph{exactly} one pair $1 \leq i,j \leq n$ such that $Y^{r,s}_{i,j} \subseteq S'$. 

For each $1 \leq r \leq k$, let $U^r_{d_r} \subseteq S'$ and, for each $1 \leq r \neq s \leq k$, let $Y^{r,s}_{p_{r,s},q_{r,s}} \subseteq S'$. This means that $$S' = Z \cup \bigcup_{r=1}^k U^r_{d_r} \cup \bigcup_{s=2}^k\;\bigcup_{r=1}^{s-1} Y^{r,s}_{p_{r,s},q_{r,s}}.$$ Since $S'$ is an $f$-critical set of $G'$, each vertex in $C$ has at least $q$ neighbors in $S'$. By our construction, the only way this is possible is if, for each set $Y^{r,s}_{p_{r,s},q_{r,s}}$, $d_r = p_{r,s}$ and $d_s = q_{r,s}$, since only in this case every vertex in $C$ has exactly $q$ neighbors in $S'$ (remember that $M_i+N_i = q$ and, for each pair $1 \leq i \neq j \leq n$, either $M_i+N_j < q$ or $M_j+N_i < q$).

Let $S = \{v_{d_r} \mid 1 \leq r \leq k\}$. Since, for each pair $1 \leq r < s \leq k$, there is a set $Y^{r,s}_{d_r,d_s}$ in $G'$, then there is an edge in $G$ between the vertices $v_{d_r}$ and $v_{d_s}$. Therefore, $S$ is a clique of $G$ of size~$k$.
\end{proofclaimfinal}
\end{proof}

\section{Algorithmic results for \texorpdfstring{$tw(G), m(f)$}{tw(G), m(f)} and \texorpdfstring{$\Delta(G)$}{Δ(G)}}
\label{sec:tractable}

Based on the results in the previous section, we can conclude that the problem is intractable even when parameterized by $m(f), \Delta(G)$ or $tw(G)$. This is surprising because many problems on graphs are in $\FPT$ when parameterized by $tw(G)$. However, in this section, we show that, if we parameterize the problem by the combinations $tw(G)+m(f)$ or $tw(G)+\Delta(G)$, the problem becomes $\FPT$.

\begin{theorem}
\label{thm:fpttwmf}
$f$-\nameref{prob:criticalSet} is in $\FPT$ when parameterized by $tw(G)+m(f)$ or by $tw(G)+\Delta(G)$.
\end{theorem}

\begin{proof}
If $m(f)=1$, $V(G)$ is the only $f$-critical set of $G$. Therefore, let us assume that $m(f) \geq 2$. The idea of this proof is to show that we can reduce an instance $(G,f,k)$ of the $f$-\nameref{prob:criticalSet} problem to an equivalent instance $(G',k')$ of the $m(f)$-\nameref{prob:criticalSet} problem such that $tw(G') \leq tw(G)+1$ in polynomial time, and then we show that we can solve $m(f)$-\nameref{prob:criticalSet} in time $O(g(m(f),tw(G')) \cdot n)$, for some function $g$ that depends only on $m(f)$ and $tw(G')$.

First, we show how to reduce an instance $(G,f,k)$ of $f$-\nameref{prob:criticalSet} to an instance $(G',k')$ of $m(f)$-\nameref{prob:criticalSet} in polynomial time, such that $tw(G') \leq tw(G)+1$. The idea is to attach gadgets to the vertices in $G$ so that each vertex behaves as if it had threshold $f(v)$, while all thresholds become equal to $m(f)$.

Let $G' = (V',E')$ be initially equal to $G=(V,E)$. For every vertex $v \in V$ such that $f(v) < m(f)$, add $m(f)-f(v)$ new neighbors to $v$ in $G'$ all of which have degree $1$ in $G'$ and let $P_v$ be the set of vertices added in this step. Then, add $m(f)-f(v)$ copies of the gadget in \Cref{fig:gadgetftoc}, and add an edge between $v$ and the vertex $w$ in the figure for each added gadget. In \Cref{fig:gadgetftoc}, $w$ has $m(f)$ neighbors with degree $1$ and $m(f) + k - 1$ neighbors with degree $m(f)+1$. Let $W_v$ be the set of all vertices $w$ adjacent to $v$ that are added in this step. Note that, because $G'[E' \setminus E]$ is a forest, $tw(G') \leq tw(G)+1$.

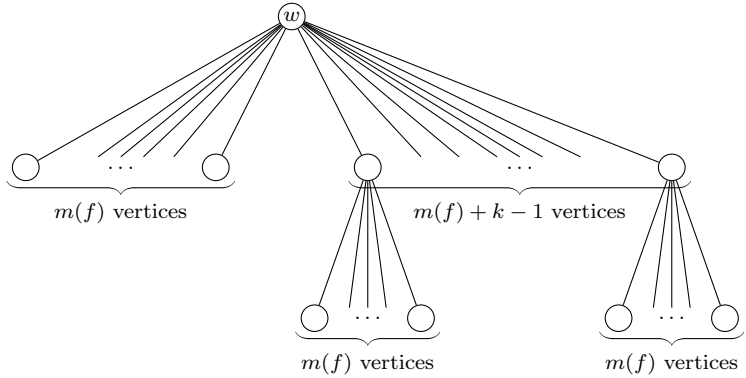
\begin{figure}[htb]
\centering
\begin{tikzpicture}[auto,font=\footnotesize]
\tikzset{vertex/.style={draw, circle, minimum size=10pt,inner sep=1pt}};
\node[vertex] (w) at (0,1) {$w$};
\node[vertex] (s1) at (-3.5,-1) {};
\node[] (r1) at (-2.25,-1) {$\ldots$};
\node[vertex] (s2) at (-1,-1) {};
\node[vertex] (n1) at (1,-1) {};
\node[] (r2) at (3,-1) {$\ldots$};
\node[vertex] (n2) at (5,-1) {};
\node[vertex] (q1) at (0.3,-3) {};
\node[] (r3) at (1,-3) {$\ldots$};
\node[vertex] (q2) at (1.7,-3) {};
\node[vertex] (l1) at (4.3,-3) {};
\node[] (r4) at (5,-3) {$\ldots$};
\node[vertex] (l2) at (5.7,-3) {};

\draw [decorate, decoration = {calligraphic brace,mirror,amplitude=5pt}] ($(s1.south)-(0.25,0)$) -- ($(s2.south)+(0.25,0)$);
\node[] () at (-2.25,-1.6) {$m(f)$ vertices};

\draw [decorate, decoration = {calligraphic brace,mirror,amplitude=5pt}] ($(n1.south)-(0.25,0)$) -- ($(n2.south)+(0.25,0)$);
\node[] () at (3,-1.6) {$m(f)+k-1$ vertices};

\draw [decorate, decoration = {calligraphic brace,mirror,amplitude=5pt}] ($(q1.south)-(0.25,0)$) -- ($(q2.south)+(0.25,0)$);
\node[] () at (1,-3.6) {$m(f)$ vertices};

\draw [decorate, decoration = {calligraphic brace,mirror,amplitude=5pt}] ($(l1.south)-(0.25,0)$) -- ($(l2.south)+(0.25,0)$);
\node[] () at (5,-3.6) {$m(f)$ vertices};

\draw[-] (w) edge (s1);
\draw[-] (w) edge (r1.north west);
\draw[-] (w) edge (r1.north);
\draw[-] (w) edge (r1.north east);
\draw[-] (w) edge ($(r1.north east) + (0.4,0)$);
\draw[-] (w) edge (s2);
\draw[-] (w) edge (n1);
\draw[-] (w) edge ($(r2.north west) - (1,0)$);
\draw[-] (w) edge ($(r2.north west) - (0.5,0)$);
\draw[-] (w) edge (r2.north west);
\draw[-] (w) edge (r2.north);
\draw[-] (w) edge (r2.north east);
\draw[-] (w) edge ($(r2.north east) + (0.5,0)$);
\draw[-] (w) edge (n2);

\draw[-] (n1) edge (q1);
\draw[-] (n1) edge (r3.150);
\draw[-] (n1) edge (r3.north);
\draw[-] (n1) edge (r3.30);
\draw[-] (n1) edge (q2);

\draw[-] (n2) edge (l1);
\draw[-] (n2) edge (r4.150);
\draw[-] (n2) edge (r4.north);
\draw[-] (n2) edge (r4.30);
\draw[-] (n2) edge (l2);
\end{tikzpicture}
\caption{\label{fig:gadgetftoc} Gadget added to $G'$ $m(f) - f(v)$ times for each $v \in V$ such that $f(v) < m(f)$.}
\end{figure}

Let $Q$ be the set of vertices with degree $1$ in $G'$. Note that no vertex in $Q$ is originally in $G$. Let $S$ be an $m(f)$-critical set of $G'$ of size $k' = k + |Q|$. It follows that $Q \subseteq S$ for any $v \in V$, since each vertex in $Q$ has degree $1$ and $m(f) \geq 2$. 

Furthermore, for any $v \in V$ and $w \in W_v$, $w$ has at most $k$ neighbors in $S \setminus Q$, since $|S| = k + |Q|$. This implies that $w$ has at least $m(f)$ neighbors in $V' \setminus S$ since $w$ has $m(f)+k$ neighbors that are not in $Q$ and at most $k$ of these neighbors are in $S$. Therefore, $w \notin S$ because, otherwise, it would change its state to $0$ at time $1$ in an $f$-reversible process on $G'$ started by $S$. 

We conclude that each vertex $v \in V$ has exactly $m(f)-f(v)$ neighbors in $V' \setminus V$ and in $S$, which are the vertices in $P_v$. Additionally, each vertex $v \in V$ has exactly $m(f)-f(v)$ neighbors that are in $V' \setminus V$ and are not in $S$, which are the vertices in $W_v$. Thus, each vertex $v \in V$ behaves as if it has threshold $f(v)$ considering only the influence of the vertices in~$V$.

Hence, letting $S \subseteq V(G)$, we have that $S$ is an $f$-critical set of $G$ of size at most $k$ if and only if $S \cup Q$ is an $m(f)$-critical set of $G'$ of size at most $k' = k + |Q|$.


Now, let us proceed to the second part of the proof. We are going to prove that $c$-\nameref{prob:criticalSet} can be solved in $O(g(c,tw(G')) \cdot n)$ time for some function $g$, where $c = m(f)$. To do this, we shall define a predicate $\mathbf{critical}(S)$ in the MSO logic with length depending only on $c$ such that $\mathbf{critical}(S)$ is satisfied by $S$ if and only if $S$ is a $c$-critical set of $G'$. Thus, we can invoke the optimization version of Courcelle's theorem~\cite{arnborg1991} using the affine function $\alpha(|S|) = |S|$ to obtain our result. 

For constants $c_1 \leq c_2$, we use the notation $\bigwedge_{i=c_1}^{c_2} P_i$ to denote the formula $P_{c_1} \land \cdots \land P_{c_2}$, where $P_k$, for each $c_1 \leq k \leq c_2$, is obtained from $P_i$ by replacing every occurrence of $i$ with $k$. We now define $\mathbf{critical}(S)$:

\begin{flalign*}
\mathbf{critical}(S) = &\left(\forall v \notin S:\; \exists x_1,\ldots,x_c \in S:\; \bigwedge\limits_{i=1}^c \adj(v,x_i) \land \bigwedge\limits_{i=1}^{c-1} \bigwedge\limits_{j=i+1}^c x_i \neq x_j \right) \land\\
&\left[\forall v \in S:\; \forall x_1,\ldots,x_c \notin S:\; \lnot \left( \bigwedge\limits_{i=1}^c \adj(v,x_i) \land \bigwedge\limits_{i=1}^{c-1} \bigwedge\limits_{j=i+1}^c x_i \neq x_j\right)\right]
\end{flalign*}

In an $f$-reversible process started by $S$, the first line of $\mathbf{critical}(S)$ ensures that every vertex that is not in $S$ has its state changed at time $t=1$, whereas the second line ensures that every vertex that is in $S$ does not have its state changed at time $t=1$. 

In conclusion, the problem is $\FPT$ when parameterized by $tw(G)$ and $m(f)$. Since $m(f) \leq \Delta(G)+1$, the problem is also in $\FPT$ when parameterized by $tw(G)$ and $\Delta(G)$.
\end{proof}

\section{Algorithmic results for the parameter \texorpdfstring{$k$}{k}}
\label{sec:tractablek}

In this section, we prove that the $f$-\nameref{prob:criticalSet} problem is in $\FPT$ when parameterized by $k$. The main difficulty of the $f$-\nameref{prob:criticalSet} problem is to balance two requirements simultaneously: the vertices in the chosen $f$-critical set $S$ must not change their state, while all remaining vertices must change to state $1$ at time $1$.

To handle this, we identify the vertices that are already forced to be in $S$, which vertices are forced to stay out of $S$, and which vertices may or may not be in $S$. This leads to a partition of $V$ into three sets. The vertices that must be in $S$ form the set $A$, the vertices that cannot be in $S$ form the set $C$, and the remaining vertices $B$ are those that may or may not be included in $S$. To form $S$, the algorithm adds all vertices from $A$ to $S$ and selects subsets of $B$ that induce connected subgraphs to be included in $S$. These subsets must be chosen so that no vertex in $S$ changes its state, and every vertex outside $S$ changes to state $1$. The remainder of this section formalizes this strategy.

Let $(G = (V,E),f,k)$ be an instance of the problem, where $n = |V|$ and $m = |E|$. A set $S  \subseteq V$ is said to be \emph{connected} if $G[S]$ is connected or $S = \varnothing$. Also, whenever we mention that a computation or algorithm runs in $\FPT$ time, we mean that it can be carried out in $O(g(k) \cdot (|V|+|E|)^c)$ time for some constant $c$ and computable function $g : \Nat \rightarrow \Nat$.

Let $S$ be an $f$-critical set of $G$ of size at most $k$, and $c$ be the $f$-reversible process on $G$ started by $S$. If there is a vertex $v \in V(G)$ such that $f(v) > k$, then $v \in S$, since its state cannot change from its initial value in $c$. Furthermore, if there is a vertex $v \in V(G)$ such that $d(v)-k \geq f(v)$, then $v \notin S$; otherwise, even if all vertices in $S$ are adjacent to $v$, its state would certainly change from $1$ at time $t=0$ to $0$ at time $t=1$ in $c$. Thus, if there is a vertex $v$ in $G$ such that $f(v) > k$ and $d(v)-k \geq f(v)$, then $(G,f,k)$ is a \textsc{NO}-instance, since otherwise we would have that $v \in S$ and $v \notin S$ for any $f$-critical set $S$ of $G$ of size at most $k$, a contradiction.

\begin{definition}
Let:
\begin{itemize}
    \item $A = \{v \in V \mid f(v) > k, d(v)-k < f(v)\}$: the set that must be a subset of any $f$-critical set of $G$ of size at most~$k$;
    \item $B = \{v \in V \mid f(v) \leq k, d(v)-k < f(v)\}$: the set of vertices with degree at most $2k-1$ that may or may not be in an $f$-critical set of $G$ of size at most~$k$; and
    \item $C = V \setminus (A \cup B) = \{v \in V \mid f(v) \leq k, d(v)-k \geq f(v)\}$: the set that must be disjoint from any $f$-critical set of $G$ of size at most~$k$.
\end{itemize}
\end{definition}

Now, let us define the set of vertices that are not guaranteed to have state $1$ at time $1$ solely due to the influence of $A$, and therefore must have their state at time $1$ explicitly controlled by the other vertices in an $f$-critical set of $G$.

\begin{definition}
Let:
\begin{itemize}
    \item $A' = \{v \in A \mid |N(v) \setminus A| \geq f(v)\}$;
    \item $B' = \{v \in B \mid |N(v) \cap A| < f(v)\}$; and
    \item $C' = \{v \in C \mid |N(v) \cap A| < f(v)\}$.
\end{itemize}
\end{definition}

Since $A$ is a subset of every $f$-critical set $S$ of $G$ of size at most $k$, the vertices in $A \setminus A', B \setminus (B' \cup S)$, and $C \setminus C'$ are guaranteed to have state $1$ at time $1$ in an $f$-reversible process on $G$ started by any superset of $A$, which includes any $f$-critical set of $G$ of size at most $k$. Therefore, we seek an $\FPT$-time algorithm that decides whether there is a pair of sets $R \subseteq B'$ and $S \subseteq B \setminus B'$, with $|R \cup S \cup A| \leq k$, satisfying the following three properties:
\begin{enumerate}[align=left, label=\textbf{Property \arabic*}:]
    \item for all $v \in S$, $|N(v) \setminus (R \cup S \cup A)| < f(v)$;
    \item for all $v \in A' \cup R$, $|N(v) 
    \setminus (R \cup S \cup A)| < f(v)$; and
    \item for all $v \in (B' \setminus R) \cup C'$, $|N(v) \cap (R \cup S \cup A)| \geq f(v)$.
\end{enumerate}

Let $c$ be an $f$-reversible process started by $R \cup S \cup A$ in $G$ such that $|R \cup S \cup A| \leq k$ and the pair $(R,S)$ satisfies the three properties. In the process $c$, Properties 1 and 2 ensure that every vertex in $R \cup S \cup A$ maintains its state, while Property 3 ensures that every vertex in $(B' \setminus R) \cup C'$ changes its state to $1$ at time $1$. Thus, Properties 1, 2 and 3 imply that $R \cup S \cup A$ is an $f$-critical set of $G$ of size at most $k$, since the vertices in $A' \setminus A$, $B \setminus (B' \cup S)$ and $C \setminus C'$ are guaranteed to have state $1$ at time $1$.

Let $S_1$ and $S_2$ be two connected sets. We say that $S_1$ and $S_2$ are adjacent if $S_1,S_2 \neq \varnothing$ and $S_1 \cup S_2$ is connected. They are non-adjacent otherwise. Any number of sets referred to as (non-)adjacent are assumed to be pairwise (non-)adjacent.

Let $R \subseteq B'$. We define the family $\mathcal{D}(R)$ of sets $S \subseteq B \setminus B'$ such that in an $f$-reversible process on $G$ started by any superset of $R \cup S \cup A$, the states of the vertices in $S$ remain unchanged at time~$1$. Because each connected component of $G[S]$ behaves independently under Property 1, we only need to consider non-adjacent connected subsets of $B \setminus B'$ to be part of $S$.
\begin{definition}
Let $R \subseteq B'$. Let $\mathcal{D}(R)$ be the family of connected sets $S$ such that $S \subseteq B \setminus B'$, $|S| \leq k - |R| - |A|$, and the pair $(R,S)$ satisfies Property~1.
\end{definition}

Notice that $\varnothing \in \mathcal{D}(R)$ for any $R \subseteq B'$ such that $|R| \leq k - |A|$, and $\mathcal{D}(R) = \varnothing$ for any $|R| > k - |A|$.

\begin{lemma}
\label[lemma]{lem:ifonlyifpt1}
Let $R \subseteq B'$ and $S \subseteq B \setminus B'$ such that $|S| \leq k - |R| - |A|$. The pair $(R,S)$ satisfies Property~1 if and only if the vertex set of each connected component of $G[S]$ is in $\mathcal{D}(R)$.
\end{lemma}

\begin{proof}
Assume that the pair $(R,S)$ satisfies Property 1. Let $Q$ be the vertex set of an arbitrary connected component of $G[S]$. Thus, $|Q| \leq |S| \leq k - |R| - |A|$ and $Q \subseteq S \subseteq B \setminus B'$. Let $v \in Q$. Since the pair $(R,S)$ satisfies Property 1, $|N(v) \setminus (R \cup S \cup A)| < f(v)$. Moreover, since a vertex in $Q$ is not adjacent to any vertex in another connected component of $G[S]$ in $G$ (otherwise they would be the same connected component), it follows that $|N(v) \setminus (R \cup Q \cup A)| = |N(v) \setminus (R \cup S \cup A)| < f(v)$. Hence, the pair $(R,Q)$ satisfies Property~1, and therefore $Q \in \mathcal{D}(R)$.

Now, let $S_1,\ldots,S_x$ be the vertex sets of the connected components of $G[S]$, and assume that each one is in $\mathcal{D}(R)$. Let $v \in S_i$ for some arbitrary $1 \leq i \leq x$. Since $(R,S_i)$ satisfies Property 1, we know that $|N(v) \setminus (R \cup S_i \cup A)| < f(v)$, which implies that $|N(v) \setminus (R \cup S \cup A)| \leq |N(v) \setminus (R \cup S_i \cup A)| < f(v)$. Therefore, since $S = S_1  \cup \ldots \cup S_x$ and $|S| \leq k - |R| - |A|$, the pair $(R,S)$ satisfies Property~1.
\end{proof}

Our algorithm must determine whether there is a set $R \subseteq B'$ and non-adjacent connected sets $S_1, \ldots, S_k$ in $\mathcal{D}(R)$ such that the sum of their sizes is at most $k - |A| - |R|$, and the pair $(R,S)$ satisfies Properties~2 and~3, where $S = S_1 \cup \ldots \cup S_k$. By \Cref{lem:ifonlyifpt1}, the pair $(R,S)$ already satisfies Property~1. Note that some of these $k$ non-adjacent connected sets may be empty, since $G[S]$ may have fewer than $k$ connected components, or even none at all.

To achieve this in $\FPT$ time, we must bound the size of $A'$, $B'$ and $C'$ in a \textsc{YES}-instance by a function of $k$. With this, our algorithm can compute all connected sets in $\mathcal{D}(R)$ for every $R \subseteq B'$ and classify each such set according to its influence on the vertices in $A',B'$ and $C'$ in a number of categories that depends only on $k$. 

Furthermore, for our algorithm to be able to select only non-adjacent sets in $\mathcal{D}(R)$ in $\FPT$ time, we must also bound the number of connected sets adjacent to any given connected set in $\mathcal{D}(R)$ by a function of $k$.

\begin{lemma}
\label[lemma]{lem:blclbounded}
If $|A| > k$ or $|A'|+|B'|+|C'| > 2k^2$, then $(G,f,k)$ is a \textsc{NO}-instance.
\end{lemma}

\begin{proof}
The set $A$ is contained in any $f$-critical set of $G$, since their states never change in any $f$-reversible process started by a set of size at most $k$. Hence, if $|A| > k$, there can be no $f$-critical set of $G$ of size at most $k$.

Now, assume that $(G,f,k)$ is a \textsc{YES}-instance. Then there are sets $R \subseteq B'$ and $S \subseteq (B \setminus B')$ such that $|R| + |S| + |A| \leq k$ and $(R,S)$ satisfies Properties~1,~2 and~3. Since every vertex in $R \cup S$ has degree at most $2k-1$, there are at most $|R \cup S| \cdot (2k-1) \leq k \cdot (2k-1)$ edges between a vertex in $R \cup S$ and a vertex in $A' \cup (B' \setminus R) \cup C'$. Therefore, $|A'| + |B' \setminus R| + |C'| \leq k \cdot (2k-1) = 2k^2 - k$; otherwise, by the pigeonhole principle, there would be a vertex in $A' \cup (B' \setminus R) \cup C'$ adjacent to no vertex in $R \cup S$, which contradicts Properties~2 or~3.

Since $R \subseteq B'$, we have $|A'| + |B' \setminus R| + |C'| = |A'| + |B'| + |C'| - |R|$. Because $|R| \leq k$, it follows that $|A'| + |B'| + |C'| \leq 2k^2$.
\end{proof}

\begin{lemma}
\label[lemma]{lem:limadj}
Let $R \subseteq B'$ be a set of size at most $k$, and let $S$ be a connected set in $\mathcal{D}(R)$. There are $O(n \cdot (6k)^{k-1})$ sets in $\mathcal{D}(R)$, which can be constructed in $O(n \cdot (6k)^{k+1})$ time. Furthermore, there are $O((6k)^{k+1})$ sets adjacent to $S$ in $\mathcal{D}(R)$, which can be constructed in $O((4k)^{2k+2})$ time.
\end{lemma}

\begin{proof}
Bollobás~\cite[Section~3.45]{bollobas2006b} showed that the number of connected induced subgraphs of a graph $H$ with exactly $k$ vertices that contain a given vertex is at most $(e \cdot (\Delta(H) - 1))^{k-1}$. Hence, the number of connected induced subgraphs of $H$ with \emph{at most} $k$ vertices that contain a given vertex is 
\begin{equation} 
\tag{E1}\label{eq:numbadjSatmk}
\sum_{i=1}^k (e \cdot (\Delta(H) - 1))^{i-1} = O((e \cdot (\Delta(H) - 1))^{k-1}).
\end{equation}

Moreover, Komusiewicz and Sommer~\cite{komusiewicz2021} showed that, for every $k \geq 1$, the number of connected induced subgraphs of a graph $H$ with \emph{exactly}~$k$ vertices is at most
\begin{equation}
\tag{E2}\label{eq:numbexack}
|V(H)|/k \cdot (e \cdot (\Delta(H) - 1))^{k-1}
\end{equation}
and presented an algorithm that constructs all connected induced subgraphs of $H$ with \emph{at most} $k$ vertices in time
\begin{equation}
\tag{E3}\label{eq:algatmk}
O((\Delta(H) + k) \cdot |V(H)|/k \cdot (e \cdot (\Delta(H) - 1))^{k-1}).
\end{equation}

Any set adjacent to $S$ must contain at least one vertex in $S$ or in $N(S)$. Since every vertex in $S \subseteq B$ has degree at most $2k-1$, we have that $|S \cup N(S)| \leq k + k (2k-1)=2k^2$. Therefore, by \Cref{eq:numbadjSatmk}, there are $$O(2k^2 (e \cdot (2k-2))^{k-1}) \subseteq O((6k)^{k+1})$$ connected sets adjacent to $S$ in $\mathcal{D}(R)$.

Let $Q$ be the set of vertices at distance at most $k$ from some vertex in $S$. Since every vertex in some connected set adjacent to $S$ in $\mathcal{D}(R)$ must be at distance at most $k$ from some vertex in $S$, all such sets are in $G[Q]$. Since $|Q| \leq \sum_{d=0}^k k \cdot (2k-1)^d = O(k \cdot (2k)^k)$, it is possible to construct all connected sets adjacent to $S$ in $\mathcal{D}(R)$ by listing all connected sets of size at most $k$ in $G[Q]$ and, for each of them, checking, in $O(k^2)$ time, whether they belong to $\mathcal{D}(R)$ and are adjacent to $S$. By \Cref{eq:algatmk}, this takes $$O((3k-1) \cdot |Q|/k \cdot (e \cdot (2k-2))^{k-1} \cdot k^2) \subseteq O((4k)^{2k+2}).$$ 

By \Cref{eq:numbexack}, for every $k \geq 1$, the number of connected induced subgraphs of $G[B \setminus B']$ with exactly $k$ vertices is at most $|B|/k \cdot (e \cdot (2k-2))^{k-1}$. Thus, the number of sets in $\mathcal{D}(R)$ is at most $$1+\sum_{i=1}^k |B|/i \cdot (e \cdot (2k-2))^{i-1} = O(n \cdot (6k)^{k-1}).$$ The sets in $\mathcal{D}(R)$ can be constructed by listing all connected sets of size at most $k$ in $G[B \setminus B']$ and, for each of them, checking, in $O(k^2)$ time, whether they belong to $\mathcal{D}(R)$. By \Cref{eq:algatmk}, this takes $O((3k-1) \cdot n/k \cdot (e \cdot (2k-2))^{k-1} \cdot k^2) \subseteq O(n \cdot (6k)^{k+1})$ time.
\end{proof}

We next classify the sets in $\mathcal{D}(R)$ for any fixed $R \subseteq B'$. Each class is defined by the size of its sets and by the influence its sets have on the vertices of $A'$, $B'$, and $C'$. These classes are important for efficiently find sets $S_1,\ldots,S_k \in \mathcal{D}(R)$ such that the pair $(R, S_1 \cup \ldots \cup S_k)$ satisfies Properties 2 and 3.

\begin{definition}
Let $Q \subseteq V$ and $x \in \Nat$. Define $\mathcal{F}$ and $\mathcal{F}(x)$ as the sets of all functions from $A' \cup B' \cup C'$ into $\Nat$ and into $\{0,1,\ldots,x\}$, respectively.
\end{definition}

\begin{definition}
Let $R \subseteq B'$, $h \in \mathcal{F}$ and $l \in \Nat$. The class $\mathcal{Z}_R(l,h)$ consists of all sets in $S \in \mathcal{D}(R)$ such that $|S| = l$ and, for every $v \in A' \cup B' \cup C'$, $|N(v) \cap S| = h(v)$.
\end{definition}

\Cref{fig:Zlf} illustrates this definition. The portion of the graph $H$ in the figure is $H[A' \cup B' \cup C' \cup S]$, that is, it may have other vertices and edges. Assume $A' = \{q\}$, $B' = \{r,s\}$, $C' = \{t\}$, and let $S = \{u,v,w\} \in \mathcal{D}(R)$. Then $S \in \mathcal{Z}_R(3,h)$, where $h(q)=1, h(r)=2, h(s)=3$ and $h(t) = 0$. It is important to note that $\mathcal{Z}_R(0,h) = \{\varnothing\}$ for $h \in \mathcal{F}(0)$. Moreover, the classes are pairwise disjoint, and $\mathcal{Z}_R(l,h) = \varnothing$ whenever $h(v) > l$ for some $v \in A' \cup B' \cup C'$.

\begin{figure}[hbt]
    \centering
    \begin{tikzpicture}[auto,font=\footnotesize]
        \tikzset{vertex/.style={draw, circle, minimum size=12pt,inner sep=1pt}};
        \node[font=\Large] () at (-3,4) {$A'$};
        \node[font=\Large] () at (0,4) {$B'$};
        \node[font=\Large] () at (3,4) {$C'$};
        \node[font=\Large] () at (3,-0.35) {$S$};
        
        \node[vertex] (q) at (-3,3) {$q$};
        \node[vertex] (r) at (-1,3) {$r$};
        \node[vertex] (s) at (1,3) {$s$};
        \node[vertex] (t) at (3,3) {$t$};

        \node[vertex] (u) at (200:2) {$u$};
        \node[vertex] (v) at (0:0) {$v$};
        \node[vertex] (w) at (-20:2) {$w$};

        \draw[color=blue] (-4,3.5) rectangle (4,2.5);
        \draw[color=blue,-] (-2,3.5) to (-2,2.5);
        \draw[color=blue,-] (2,3.5) to (2,2.5);
        \draw[color=blue]  (0,-0.35) ellipse (2.5 and 1.15);

        \draw[-] (u) to (v) to (w) to (u);
        \draw[-] (v) to (q);
        \draw[-] (u) to (r);
        \draw[-] (v) to (r);
        \draw[-] (u) to (s);
        \draw[-] (v) to (s);
        \draw[-] (w) to (s);
    \end{tikzpicture}
    \caption{\label{fig:Zlf}Let $A' = \{q\}$, $B' = \{r,s\}$, $C' = \{t\}$, and $S = \{u,v,w\} \in \mathcal{D}(R)$. Then $S \in \mathcal{Z}_R(3,h)$, where $h(q)=1, h(r)=2, h(s)=3$ and $h(t) = 0$.}
\end{figure}
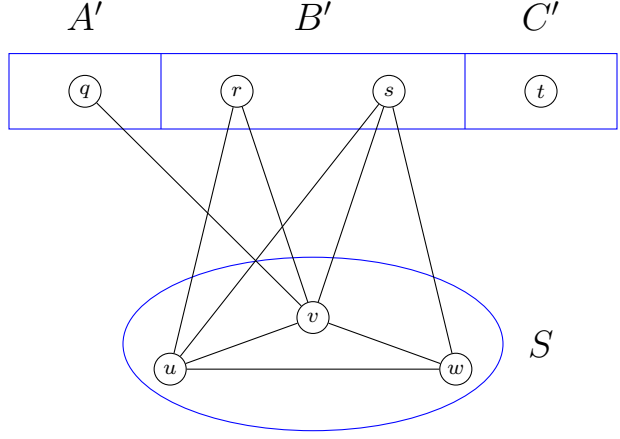

\begin{lemma}
\label[lemma]{lem:dr}
Let $R \subseteq B'$. Given $\mathcal{D}(R)$, $A'$, $B'$, and $C'$, all sets in $\mathcal{D}(R)$ can be placed into their respective classes in $O(|\mathcal{D}(R)| \cdot (|A' \cup B' \cup C'| + k^2))$ time.
\end{lemma}

\begin{proof}
The classification procedure works as follows: for each set $S \in \mathcal{D}(R)$, it initializes a function $h \in \mathcal{F}(0)$. Then, for each $v \in S$ and $w \in N(v)$, it increments $h(w)$ whenever $w \in A' \cup B' \cup C'$. After that, it places $S$ in $\mathcal{Z}_R(|S|,h)$. Since $S \in \mathcal{D}(R)$, we have that $|S| \leq k$, and $d(v) \leq 2k-1$ for every $v \in S$. Therefore, the procedure runs in $O(|\mathcal{D}(R)| \cdot (|A' \cup B' \cup C'| + k^2))$ time.
\end{proof}

The general idea of our algorithm is that, for each $R \subseteq B'$, it checks whether it is possible to construct a set $S \subseteq B \setminus B'$ by taking the union of non-adjacent sets in the classes to ensure the pair $(R,S)$ satisfies all three properties. 

To do this, our algorithm first fixes $R \subseteq B'$. This establishes a baseline contribution to the number of neighbors for each vertex $v \in A' \cup B' \cup C'$ (namely, $|N(v) \cap A|$ and $|N(v) \cap R|$). The algorithm must then select $k$ sets $S = S_1 \cup \dots \cup S_k$ to provide the remaining neighbors required for $(R,S)$ to satisfy Properties~2 and~3. 

This is achieved by selecting each set $S_i$ from specific classes $\mathcal{Z}_R(l_i, h_i)$. These classes are crucial because they capture the exact contribution $h_i(v)$ that any choice of $S_i$ from that class would make to the neighborhood of $v$. The algorithm searches for a combination of $k$ such classes whose cumulative contribution $\sum_{i=1}^k h_i(v)$ is sufficient to satisfy the threshold inequalities for both Properties~2 and~3 simultaneously. The non-adjacency of the chosen sets ensures that their contributions do not overlap, allowing the total contribution of $S = S_1 \cup \dots \cup S_k$ to be computed correctly.

Additionally, the chosen sets must be in $\mathcal{D}(R)$ and be non-adjacent to ensure Property~1 by \Cref{lem:dr}. To select non-adjacent sets $S_1,\ldots,S_k$ in $\mathcal{D}(R)$, each $S_i$ in its specific class $\mathcal{Z}_R(l_i,h_i)$, the algorithm constructs a graph $H$. In $H$, each vertex represents a set in $\mathcal{D}(R)$, and two vertices are adjacent if the corresponding sets are either in the same class or adjacent in the original graph. The algorithm then searches for an independent set in $H$ of size $k$, that is, a set of $k$ vertices that are pairwise non-adjacent in $H$. $H$ is constructed in a way that, if there is an independent set $Q$ in $H$ of size $k$, then the connected sets represented by the vertices in $Q$ are exactly the non-adjacent sets $S_1,S_2,\ldots,S_k$ such that $S_i \in \mathcal{Z}_R(l_i,h_i)$ for every $1 \leq i \leq k$.

\begin{algorithm}[H]
\caption{\label{alg:fptk}Algorithm for the $f$-\nameref{prob:criticalSet} problem}
\SetAlgoLined
\SetKwInOut{Input}{Input}
\SetKw{Or}{\itshape or}
\BlankLine
\Input{A graph $G$, a threshold function $f$, and an integer $k$.}
\BlankLine
Compute the sets $A$, $B$, $C$, $A'$, $B'$, and $C'$\\
\If{$|A| > k$ \Or $|A'|+|B'|+|C'| > 2k^2$ \Or $\exists v$ s.t. $f(v) > k$ and $d(v) - k \geq f(v)$}{
    \Return NO
}
Compute $|N(v) \cap A|$ for each $v \in A' \cup B' \cup C'$\\
\ForEach{$R \subseteq B'$ of size at most $k-|A|$}{
    Compute $|N(v) \cap R|$ for each $v \in A' \cup B' \cup C'$\\
    Compute $B' \setminus R$\\    
    Compute $\mathcal{D}(R)$ and classify its sets into their respective classes $\mathcal{Z}_R(l,h)$\\
    \ForEach{$l_1,l_2,\ldots,l_k \in \Nat$ such that $l_1+l_2+\ldots+l_k \leq k-|A|-|R|$}{
        \ForEach{$h_1 \in \mathcal{F}_R(l_1)$, $h_2 \in \mathcal{F}_R(l_2),\ldots,h_k \in \mathcal{F}_R(l_k)$ such that\\
        \nonl \quad 1. for each $v \in A' \cup R$, $\sum\limits_{i=1}^k h_i(v) \geq d(v)-f(v)+1 - |N(v) \cap A| - |N(v) \cap R|$; and\\
        \nonl \quad 2. for each $v \in (B' \setminus R) \cup C'$, $\sum\limits_{i=1}^k h_i(v) \geq f(v) - |N(v) \cap A| - |N(v) \cap R|$\\\nonl \quad}{
            Construct the graph $H$, where each vertex represents one set in $\mathcal{Z}_R(l_i, h_i)$ for each $1 \leq i \leq k$. Two vertices are neighbors if they represent sets that belong to the same class $\mathcal{Z}_R(l_i, h_i)$ or that are adjacent to each other\\
            \If{there is an independent set in $H$ of size $k$}{
                \Return YES
            }
        }
    }
}
\Return NO
\end{algorithm}

Since $S_1,S_2,\ldots,S_k$ are non-adjacent sets, Items~1 and~2 of the \textbf{foreach} in Line~10 ensure that the pair $(R,S_1 \cup \ldots \cup S_k)$ satisfies Properties~2 and~3, respectively. Furthermore, because the sets $S_1,S_2,\ldots,S_k$ belong to $\mathcal{D}(R)$ and are non-adjacent, Property~1 is also satisfied.

It is important to highlight that, when constructing the vertices of $H$ in Line~11, even if $\mathcal{Z}_R(l_r, h_r) = \mathcal{Z}_R(l_s, h_s)$ for some $r \neq s$, the algorithm still creates separate vertices for each set in $\mathcal{Z}_R(l_r, h_r)$ and $\mathcal{Z}_R(l_s, h_s)$. In this case, there is no edge between a vertex representing a set in $\mathcal{Z}_R(l_r, h_r)$ and a vertex representing a set in $\mathcal{Z}_R(l_s, h_s)$ unless the corresponding sets are adjacent. That is, the algorithm treats both classes as if they are distinct when constructing $H$.

\begin{lemma}
\label[lemma]{lem:time}
\Cref{alg:fptk} runs in $O(n \cdot k^{\Theta(k^3)})$ time.
\end{lemma}

\begin{proof}
Let $n = |V(G)|$. Lines~1 through~4 of \Cref{alg:fptk} run in $O(n)$ time. From Line~4 onward, observe that $|A'| \leq |A| \leq k$ and $|A'|+|B'|+|C'| \leq 2k^2$. There are $O((2k^2)^k)$ subsets $R$ of $B'$ of size at most $k$, and it takes $O(k \cdot (2k^2)^k)$ time to generate all of them.

Line~6 runs in $O(nk)$ time. Line~7 runs in $O(n)$ time. By \Cref{lem:limadj,lem:dr}, Line~8 runs in $O(n \cdot (6k)^{k+1} + n \cdot (6k)^{k-1} \cdot k^2) \subseteq O(n \cdot (6k)^{k+1})$ time.

There are at most $\binom{2k-1}{k} < ((2k-1)/k \cdot e)^k = O(6^k)$ sequences of non-negative integers $l_1,\ldots,l_k$ such that $l_1+\ldots+l_k \leq k-|A|-|R|$. Moreover, these sequences can be generated in $O(k \cdot 6^k)$ time.

Since every function $h \in \mathcal{F}(l)$ with $l \leq k$ satisfies $h(v) \in \{0,1,\ldots,k\}$, there are at most $(k+1)^{(|A'|+|B'|+|C'|) \cdot k} = O(k^{2k^3})$ functions $h_1,\ldots,h_k$ that satisfy the conditions of the algorithm. It is possible to generate them in $O(2k^3 \cdot k^{2k^3}) = O(k^{2k^3+3})$, provided that we already know the values $|N(v) \cap A|$ and $|N(v) \cap R|$ for each $v \in A' \cup B' \cup C'$.

Fix a set $R$, integers $l_1,\ldots,l_k$, and functions $h_1,\ldots,h_k$ that satisfy the conditions of \Cref{alg:fptk}. By \Cref{lem:limadj,lem:dr}, the vertex set of the graph $H$ constructed in Line~11 has size $O(n \cdot (6k)^{k-1})$ and can be constructed in $O(n \cdot (6k)^{k+1})$ time. Its edge set can be constructed in $O(n \cdot (6k)^{k-1} \cdot (4k)^{2k+2}) \subseteq O(n \cdot (6k)^{3k+1})$ time.

The \textsc{Maximum Independent Set} problem asks, given a graph $G$ and an integer $k$, whether there is an independent set of $G$ of size $k$. It is a classical $\W[1]$-complete problem when parameterized solely by $k$. However, it has a well-known linear-time computable kernel of size $O((k-1) \cdot (\Delta(G)+1))$. This follows from the fact that every graph with at least $(k-1) \cdot (\Delta(G)+1) + 1$ vertices must contain an independent set of size at least $k$. Furthermore, there is a bounded search tree algorithm that solves the problem and runs in $O((\Delta(G)+1)^k \cdot n)$ time \cite{cygan2015}.

Computing the kernel, followed by running the bounded search tree algorithm on it, gives an $O(n + (\Delta(G)+1)^{k+1}\cdot k)$ time algorithm for \textsc{Maximum Independent Set}. Note that, by \Cref{lem:limadj}, the maximum degree of $H$ is $O((4k)^{2k+2})$. Thus, given $H$, it is possible to determine whether there is an independent set in $H$ of size $k$ in $O(n \cdot (6k)^k + (4k)^{2k^2+k+1})$ time.

The costs of computing and enumerating the combinatorial structures of \Cref{alg:fptk}, except those in Lines~11 and~12, are dominated by later terms and are absorbed by the $O$-notation. Multiplying the number of combinatorial choices in Lines~5,~9, and~10 by the sum of the costs of Lines~11 and~12 yields the time complexity $O((2k^2)^k \cdot 6^k \cdot k^{2k^3+3} \cdot (n \cdot (6k)^{3k+1} + n \cdot (6k)^k + (4k)^{2k^2+k+1})) \subseteq O(n \cdot k^{\Theta(k^3)})$ for \Cref{alg:fptk}.
\end{proof}

\begin{theorem}
\label{thm:fptk}
$f$-\nameref{prob:criticalSet} can be decided in $O(n \cdot k^{\Theta(k^3)})$ time.
\end{theorem}

\begin{proof}
The correctness of \Cref{alg:fptk} follows from \Cref{lem:ifonlyifpt1} and the definition of the classes $\mathcal{Z}_R(l_i,h_i)$. Its time complexity follows from \Cref{lem:time}.
\end{proof}

\Cref{alg:fptk} is correct and runs in $\FPT$ time for the sole parameter $k$, but it is rather complex. When parameterized by $k+m(f)$ or $k+\Delta(G)$, there is, in fact, a much simpler $\FPT$-time algorithm for the $f$-\nameref{prob:criticalSet} problem.

\begin{lemma}
\label[lemma]{lemma:lowerbound}
Let $G$ be a graph and $f:V(G) \rightarrow \Nat$ a threshold function. If $S$ is an $f$-critical set of $G$, then $|S| \cdot m(f) \geq n$.
\end{lemma}

\begin{proof}
Let $S' = V(G) \setminus S$ and $F$ be the set of edges between a vertex in $S$ and a vertex in $S'$. On the one hand, $|F| \geq |S'| = n - |S|$, because every vertex in $S'$ must have at least one neighbor in $S$ to change its state at time $1$. On the other hand, $|F| \leq |S| \cdot (m(f)-1)$, since, otherwise, by the pigeonhole principle, there would be a vertex in $S$ with at least $m(f)$ neighbors in $S'$ and thus the state of this vertex would change to $0$ at time $1$. Therefore, $|S| \cdot (m(f)-1) \geq n - |S|$, which implies that $|S| \cdot m(f) \geq n$.
\end{proof}

\begin{theorem}
\label{thm:fptkmf}
$f$-\nameref{prob:criticalSet} can be decided in $O(n + (k \cdot m(f))^{k+2})$ and $O(n +(k \cdot \Delta(G))^{k+2})$.
\end{theorem}

\begin{proof}
Let $\Delta = \Delta(G)$. If $k \cdot m(f) < n$, then by \Cref{lemma:lowerbound} we are dealing with a \textsc{NO}-instance. If $k \cdot m(f) \geq n$, a brute-force algorithm solves the problem and runs in $O(n + (k \cdot m(f))^{k+2})$ time, since for each set $S \subseteq V$ of size at most $k$, it is possible to check whether $S$ is an $f$-critical set of $G$ in $O(n+m) \subseteq O((k \cdot m(f))^2)$. Furthermore, since $m(f) \leq \Delta+1$, the problem can be decided in $O(n +(k \cdot \Delta)^{k+2})$.
\end{proof}
\section{Final remarks}
\label{sec:finalremarks}

In this paper, we studied the parameterized complexity of the $f$-\nameref{prob:criticalSet} problem with respect to the parameters $k$, $tw(G)$, $m(f)$, and $\Delta(G)$. Our results reveal a precise boundary in the parameterized complexity of the reversible $f$-\nameref{prob:criticalSet} problem for these parameters.

We proved that the problem is $\NP$-complete even when restricted to planar subcubic bipartite graphs with maximum threshold~$2$. Moreover, we showed that it is $\W[1]$-hard when parameterized by the treewidth $tw(G)$, which is surprising, as many problems are $\FPT$ in this parameter. On the positive side, we established that the problem becomes fixed-parameter tractable when parameterized by the combined parameters $tw(G)+m(f)$ and $tw(G)+\Delta(G)$. Furthermore, we showed that the problem is in $\FPT$ when parameterized by $k$, which contrasts with the irreversible-process setting, where the same problem is $\W[2]$-hard due to the equivalence between $1$-critical sets in irreversible processes and dominating sets.

These results provide a better understanding of the complexity landscape of problems involving $f$-reversible processes, highlighting that identifying critical sets of size at most $k$ behaves differently from the classical conversion set problem: the former is fixed-parameter tractable when parameterized by $k$, while the latter is $\W[1]$-hard under the same parameter.

An interesting open question concerns the complexity of the $f$-\nameref{prob:criticalSet} problem parameterized by the treewidth $tw(G)$ in the irreversible setting, which is listed as open in \Cref{tab:results}. The $\FPT$ algorithms obtained for the combined parameters $tw(G)+m(f)$ and $tw(G)+\Delta(G)$ follow mainly from MSOL formulations and Courcelle's theorem. Although our results focus on reversible processes, similar logical formulations appear possible for irreversible processes, suggesting that these techniques may also be applicable in that setting. Moreover, the reduction used to prove $\W[1]$-hardness when the problem is parameterized by $tw(G)$ seems adaptable to the irreversible case with only minor modifications.

We believe the reversible model deserves further investigation, since it more accurately captures real-world influence dynamics than the irreversible model in many situations. Although we provided an $\FPT$ algorithm for the parameter $k$, the dependence on $k$ is quite high ($O(k^{\Theta(k^3)})$) due to the cubic size of some of the combinatorial structures it has to branch. Despite the fact that this high dependence on $k$ is mostly due to the cubic size of the combinatorial choices involved, we believe there is room for improvement, especially since the only known ETH-based lower bound for the problem is $poly(n+m) \cdot 2^{\Theta(k)}$ due to the $\NP$-completeness proof in~\cite{lima2018}. So, we leave as an open problem whether the $f$-\nameref{prob:criticalSet} admits a polynomial kernel when parameterized by $k$, or whether such kernels can be ruled out under standard assumptions such as ETH.

Another direction is to determine the parameterized complexity of the problem on sparse graph classes (e.g., planar graphs). Finally, it would be interesting to investigate the existence of approximation algorithms for minimizing $f$-critical sets in reversible processes, which would reveal whether the problem admits efficient global optimization or inherently resists approximation. 

\section*{Acknowledgments}

This study was financed by the Conselho Nacional de Desenvolvimento Científico e Tecnológico (CNPq) -- Universal [422912/2021-2].

\bibliographystyle{elsarticle-num} 
\bibliography{refs}

\end{document}